\documentclass[orivec]{llncs}

\usepackage{amsmath,amssymb,latexsym,stmaryrd,mathtools}
\usepackage{xspace}
\usepackage[usenames,dvipsnames]{xcolor} 
\usepackage[all]{xy}
\usepackage{graphicx}
\usepackage{wrapfig} 
\usepackage{tikz}
\usepackage{multicol}
\usepackage{placeins}
\usepackage{dashrule} 
\usepackage[capitalize]{cleveref}
\crefformat{equation}{(#2#1#3)}
\usepackage{tabularx}
\usepackage{wrapfig}
\usepackage{caption}
\usepackage{subcaption}

\usepackage{listings}

\newcommand{\mrm}[1]{\mathrm{#1}}
\newcommand{\mi}[1]{\mathit{#1}}

\newcommand{\mt}[1]{\mathtt{#1}}

\newcommand{\LST}[1]{[{#1}]}
\newcommand{\defEq}{{:=}}
\newcommand{\synDefEq}[0]{::=}

\newcommand{\interp}[1]{\llbracket{#1}\rrbracket}

\newcommand{\nats}{\mathbb{N}}
\newcommand{\ints}{\mathbb{Z}}

\newcommand{\lstPrepend}[0]{::}

\newcommand{\emptyList}[0]{[\,]}

\newcommand{\TT}[0]{\mi{tt}}
\newcommand{\FF}[0]{\mi{ff}}

\newcommand{\optional}[1]{{#1}_{\bot}}

\newcommand{\component}[2]{{#1}.{#2}}

\newcommand{\Inference}[2]{\begin{array}{@{}c@{}}#1\\[0em]\hline\\[-0.9em]#2\\ \end{array}}

\newcommand{\wildcard}[0]{{\_}}

\newcommand{\powerSet}[1]{\mathcal{P}({#1})} 

\newcommand{\premise}[2]{({#1},{#2})}
\newcommand{\concl}[2]{({#1},{#2})}

\newcommand{\RULE}[2]{\mi{rule}~{[#2]}~{#1}} 
\newcommand{\DERIV}[1]{\mi{deriv}({#1})}

\newcommand{\INFER}[3]{\mi{infer}^{#3}({#1},{#2})}

\newcommand{\specValid}[1]{\mi{valid}({#1})} 

\newcommand{\rConfigFor}[3]{\mi{res}^{#3}({#1},{#2})}
\newcommand{\verifySpec}[1]{\mi{verif}({#1})}

\newcommand{\setsOfRConfigsMeta}[0]{P}

\newcommand{\specsMeta}[0]{\Phi}

\newcommand{\specLe}[0]{\preceq}
\newcommand{\fSpecLe}[2]{{#1}\specLe{#2}} 

\newcommand{\pgVar}[1]{\mt{#1}}
\newcommand{\pgVarVal}[2]{{#1}(\pgVar{#2})} 

\newcommand{\configs}[0]{C}
\newcommand{\rConfigs}[0]{R}

\newcommand{\configsMeta}[0]{c}
\newcommand{\rConfigsMeta}[0]{r}

\newcommand{\whConfig}[2]{\langle {#1}, {#2} \rangle}

\newcommand{\whStatesMeta}[0]{\sigma}
\newcommand{\whStates}[0]{\Sigma}

\newcommand{\valsMeta}[0]{v}

\newcommand{\vars}{\mi{Var}}

\newcommand{\varsMeta}[0]{x}

\newcommand{\op}[0]{\mi{op}}
\newcommand{\aop}[0]{\mi{aop}}
\newcommand{\cop}[0]{\mi{cop}}

\newcommand{\stmts}[0]{Stmt}

\newcommand{\exprsMeta}{e}

\newcommand{\stmtsMeta}[0]{S}

\newcommand{\trueLit}{\mathsf{true}}
\newcommand{\falseLit}{\mathsf{false}}

\newcommand{\SKIP}[0]{\mathsf{skip}}
 
\newcommand{\ASSG}[2]{{#1}:={#2}}
\newcommand{\SEQ}[2]{{#1};{#2}} 
\newcommand{\IF}[3]{\mathsf{if}~{#1}~\mathsf{then}~{#2}~\mathsf{else}~{#3}}  
\newcommand{\WHILE}[2]{\mathsf{while}~{#1}~\mathsf{do}~{#2}} 

\newcommand{\whAEval}[1]{\mathcal{A}\llbracket{#1}\rrbracket}
\newcommand{\whBEval}[1]{\mathcal{B}\llbracket{#1}\rrbracket}

\newcommand{\aExpsMeta}[0]{a}
\newcommand{\bExpsMeta}[0]{b} 
\newcommand{\progsMeta}[0]{\rho} 

\newcommand{\pdsMeta}[0]{w}

\newcommand{\arrRef}[2]{{#1}[{#2}]} 

\newcommand{\funCall}[3]{{#1}({#2})\to{\LST{#3}}} 
\newcommand{\numeralsMeta}[0]{n}

\newcommand{\fidsMeta}[0]{f}

\newcommand{\varDecl}[1]{\mathsf{var}~{#1}}
\newcommand{\arrDecl}[2]{\mathsf{arr}~{#1}[{#2}]}

\newcommand{\addExpr}[2]{{#1}+{#2}}
\newcommand{\subExpr}[2]{{#1}-{#2}}
\newcommand{\mulExpr}[2]{{#1}*{#2}}
\newcommand{\divExpr}[2]{{#1}\,/\,{#2}} 
\newcommand{\andExpr}[2]{{#1}{\,\&\!\&\,}{#2}}
\newcommand{\notExpr}[1]{{!{#1}}}

\newcommand{\arrsMeta}[0]{X}

\newcommand{\storesMeta}[0]{s}

\newcommand{\locBoundsMeta}[0]{\iota} 

\newcommand{\intsMeta}[0]{z}
\newcommand{\locsMeta}[0]{\ell}

\newcommand{\ewhAEval}[1]{\mathcal{A}\llbracket{#1}\rrbracket}
\newcommand{\ewhBEval}[1]{\mathcal{B}\llbracket{#1}\rrbracket}

\newcommand{\ewhNEval}[1]{\mathcal{N}\llbracket{#1}\rrbracket}

\newcommand{\ewhStatesMeta}[0]{\sigma}

\newcommand{\ewhState}[2]{({#1}, {#2})} 

\newcommand{\ewhConfig}[3]{\langle {#1},{#2} \rangle_{#3}} 
\newcommand{\ewhRConfig}[2]{({#1},{#2})}

\newcommand{\storeUpd}[3]{{#1}[{#2}\mapsto {#3}]}
\newcommand{\heapUpd}[3]{{#1}[{#2}\mapsto {#3}]} 

\newcommand{\callIniSt}[4]{\mi{call}\mbox{-}\mi{ini}({#1},{#2},{#3},{#4})} 
\newcommand{\callFinSt}[4]{\mi{call}\mbox{-}\mi{fin}({#1},{#2},{#3},{#4})}  

\newcommand{\listsOfNamesMeta}[0]{\mi{ws}}
\newcommand{\listsOfValsMeta}[0]{\mi{vs}}
\newcommand{\listsOfVarsMeta}[0]{\mi{xs}}

\newcommand{\arrs}[0]{\mi{Arr}} 

\newcommand{\varOrArrFragMeta}[0]{u}
\newcommand{\arrFrag}[3]{{#1}^{#3}_{#2}} 

\newcommand{\fVarsMeta}[0]{x}
\newcommand{\fExprsMeta}[0]{e}
\newcommand{\fNumConstsMeta}[0]{n}

\newcommand{\trueConst}[0]{\mathsf{true}}
\newcommand{\falseConst}[0]{\mathsf{false}} 

\newcommand{\LETRECIN}[3]{\mathsf{letrec}~{#1}={#2}~\mathsf{in}~{#3}}
\newcommand{\IFTHENELSE}[3]{\mathsf{if}~{#1}~\mathsf{then}~{#2}~\mathsf{else}~{#3}}
\newcommand{\LAM}[2]{\lambda {#1}.{#2}}

\newcommand{\FADD}[2]{{#1}+{#2}}
\newcommand{\FSUB}[2]{{#1}-{#2}}
\newcommand{\FMUL}[2]{{#1}*{#2}}
\newcommand{\FDIV}[2]{{#1}/{#2}}

\newcommand{\FEQ}[2]{{#1}={#2}}
\newcommand{\FLT}[2]{{#1}<{#2}}

\newcommand{\FNEG}[1]{\lnot{#1}}
\newcommand{\FAND}[2]{{#1} \land {#2}}

\newcommand{\FAPP}[2]{{#1}\,{#2}}
\newcommand{\NIL}[0]{\mathsf{nil}}
\newcommand{\CONCAT}[2]{{#1}::{#2}}
\newcommand{\LCASE}[3]{\mathsf{listcase}~{#1}~\mathsf{of}~({#2},{#3})}

\newcommand{\intCfmsMeta}[0]{\mi{icf}}
\newcommand{\boolCfmsMeta}[0]{\mi{bcf}}
\newcommand{\funCfmsMeta}[0]{\mi{fcf}}
\newcommand{\lstCfmsMeta}[0]{\mi{lcf}}
\newcommand{\cfmsMeta}[0]{\mi{cf}}

\newcommand{\interpICfm}[1]{\interp{#1}_{\mrm{icf}}}
\newcommand{\interpBCfm}[1]{\interp{#1}_{\mrm{bcf}}} 

\newcommand{\specFac}[0]{\specsMeta_{\mrm{fac}}} 
\newcommand{\stmtFac}[0]{S_{\mrm{fac}}} 

\newcommand{\stmtWh}[0]{S_{\mrm{wh}}} 
  
\newcommand{\FACT}[1]{{#1}!}

\newcommand{\Pm}[1]{\setsOfRConfigsMeta_{#1}}
\newcommand{\Pmfac}[2]{\setsOfRConfigsMeta'_{{#1},{#2}}}

\newcommand{\mergeFId}[0]{\mathsf{merge}}

\newcommand{\SArr}[0]{\mathsf{S}}
\newcommand{\TArr}[0]{\mathsf{T}}

\newcommand{\mVar}[0]{\pgVar{m}}
\newcommand{\nVar}[0]{\pgVar{n}} 
\newcommand{\iVar}[0]{\pgVar{i}}
\newcommand{\jVar}[0]{\pgVar{j}}
\newcommand{\kVar}[0]{\pgVar{k}}

\newcommand{\iVal}[0]{i}
\newcommand{\jVal}[0]{j}
\newcommand{\mVal}[0]{m}
\newcommand{\nVal}[0]{n}
\newcommand{\kVal}[0]{k}

\newcommand{\lVal}[0]{l}
\newcommand{\hVal}[0]{h} 

\newcommand{\aExpl}[0]{a_{\mrm{l}}}
\newcommand{\aExpm}[0]{a_{\mrm{m}}}
\newcommand{\aExph}[0]{a_{\mrm{h}}} 

\newcommand{\mergeStmt}[0]{S_{\mrm{mg}}}
\newcommand{\mgLoopStmt}[0]{S_{\mrm{wh}}}
\newcommand{\tailLoopStmt}[2]{S_{{#1},{#2}}}

\newcommand{\mergeProg}[0]{\progsMeta_{\mrm{mg}}} 
\newcommand{\mergeSortProg}[0]{\progsMeta_{\mrm{ms}}}

\newcommand{\specMSort}[0]{\specsMeta_{\mrm{mga}}}

\newcommand{\sepArrFrag}[3]{\mi{sep}({#1},{#2},{#3})}
\newcommand{\listOfArrFrag}[2]{(\!|{#1}|\!)_{#2}} 
\newcommand{\occ}[1]{\mi{occ}\,{#1}} 
\newcommand{\sorted}[1]{\mi{sorted}\,{#1}}

\newcommand{\occFuncsMeta}[0]{h} 
\newcommand{\occAdd}[2]{{#1}\oplus{#2}}

\newcommand{\lParam}[0]{l}

\newcommand{\preserved}[3]{\LST{{#3}}^{#2}_{#1}}

\newcommand{\yarrsMeta}[0]{Y}

\newcommand{\mergeExpr}[0]{\fExprsMeta_{\mrm{mg}}}
\newcommand{\ifExpr}[0]{\fExprsMeta_{\mrm{if}}}
\newcommand{\lcaseExpr}[0]{\fExprsMeta_{\mrm{lcase}}} 
\newcommand{\mergeVar}[0]{\pgVar{merge}}
\newcommand{\xFVar}[0]{\pgVar{x}}
\newcommand{\xpFVar}[0]{\pgVar{x'}}
\newcommand{\iFVar}[0]{\pgVar{i}}
\newcommand{\ipFVar}[0]{\pgVar{i'}}
\newcommand{\rFVar}[0]{\pgVar{r}}
\newcommand{\rpFVar}[0]{\pgVar{r'}}

\newcommand{\specMGList}[0]{\specsMeta_{\mrm{mgl}}}
\newcommand{\listOfLstCfm}[1]{\langle\!|{#1}|\!\rangle}
\newcommand{\listsOfIntsMeta}[0]{\mi{zs}}

\newcommand{\fSubst}[3]{{#1}[{#3}/{#2}]} 

\usepackage[most]{tcolorbox}

\usepackage{fancybox}

\renewcommand{\lastandname}{\unskip,}

\title{Reasoning about Iteration and Recursion Uniformly based on Big-step Semantics}
\author{}
\institute{}
\author{
  Ximeng Li\inst{1,3}, Qianying Zhang\inst{1,2}, Guohui Wang\inst{2,3}, Zhiping Shi\inst{1}, Yong Guan\inst{3}
}
 
\institute{
 Beijing Key Laboratory of Electronic System Reliability and 
 Prognostics
 \and
 Beijing Engineering Research Center of High Reliable Embedded System
 \and 
 Beijing Advanced Innovation Center for Imaging Theory and Technology
\\[1ex]
 Capital Normal University, Beijing, China
}

\date{} 

\begin{document}

\maketitle

\vspace*{-2ex}
\begin{abstract}
  A reliable technique for deductive program verification should be
  proven sound with respect to the semantics of the programming
  language. For each different language, the construction of a
  separate soundness proof is often a laborious undertaking. In
  language-independent program verification, common aspects of
  computer programs are addressed to enable sound reasoning for all
  languages. In this work, we propose a solution for the sound
  reasoning about iteration and recursion based on the big-step
  operational semantics of any programming language. We give inductive
  proofs on the soundness and relative completeness of our reasoning
  technique.  We illustrate the technique at simplified programming
  languages of the imperative and functional paradigms, with diverse
  features. We also mechanism all formal results in the Coq proof
  assistant.
  
  % For a reliable deductive program verification, the verification
  % procedure should be proven sound with respect to the formal
  % semantics of the programming language. For each different language,
  % a separate proof is needed, which is a laborious undertaking. In
  % language-independent program verification, common aspects of
  % computer programs are addressed to enable sound reasoning for all
  % languages. In this work, we propose a solution for the sound
  % reasoning about iteration and recursion based on the big-step
  % operational semantics of any programming language. We give inductive
  % proofs on the soundness and relative completeness of our technique.
  % We illustrate the technique at simplified programming languages of
  % the imperative and functional paradigms. We also mechanism all
  % formal results in the Coq proof assistant.
  
  % Existing language-independent developments supporting
  % verification in a proof assistant mostly targets languages with a
  % small-step semantics. 
\end{abstract}

%%% Local Variables:
%%% mode: latex
%%% TeX-master: "main"
%%% End:

\section{Introduction}

% Deductive program verification has been a subject of active
% investigation for decades.
% %
% This area is still receiving a great amount of attention from the
% formal methods research community, as well as the industry.
% % 
% The drive for further research comes from the involvement of multiple
% non-trivial aspects in deductive program verification:
% %
% soundness, which reduces the trust base of the verification to the
% formal semantics of the programming language; 
% %
% vesatility, which enables the reasoning about diverse language features; 
% %
% automation, which minimizes the required amount of human intervention,
% etc.

% The problem of verifying computer programs has been the subject of
% active investigation for decades.
% %
% This problem is still receiving a great amount of attention from the
% researchers of the formal methods community, as well as the
% practitioners in the industry.
% %
% The problem is non-trivial, because of the involvment of multiple
% aspects: 
% %
% soundness that reduces the trust base of the verification to the
% formal semantics of the programming language,
% %
% flexility that enables the verification of deep properties,
% %
% vesatility that enables the reasoning about diverse language features,
% %
% automation that minimize the required amount of human intervention, etc. 

% Deductive program verification has been a subject of extensive 
% investigation for decades.
% %
% This area is still receiving a great amount of attention from the
% formal methods research community and the industry. 
%
It is commonly accepted that a reliable technique for deductive
program verification should be designed with the formal semantics of
the programming language as foundation.
With the formal semantics used as axioms, a mathematical proof of a
desired property for the target program can be constructed.
Direct program proofs based on operational semantics are often
cumbersome.
Due to language constructs that may incur unbounded program behavior,
inductive proofs along the structure of semantic derivations are
expected~\cite{NielsonNielson2007}. 

% One of the complicating factors is constructs with potentially
% unbounded behavior such as loops and recursive functions.
% %
% For instance, it takes an explicit induction to establish the
% correctness of each loop, when the verification is based on an
% operational semantics~\cite{NielsonNielson2007}.

An established method for simplifying the verification is by devising a
program logic (e.g.,~\cite{Hoare1969,Reynolds2002}) for the
programming language.
Program logics effectively reduce the burdens in dealing with many
aspects of the verification, such as the reasoning about loops,
recursive function calls, memory layout of objects, concurrency, etc.
% Program logics are effective in reducing the burdens of dealing with
% many aspects in program verification, such as the reasoning about
% loops, recursive function calls, memory layout of objects,
% concurrency, etc.
%
The effectiveness of program logics has been demonstrated by powerful
tools (e.g.,~\cite{VCC,Appel2011,AhrendtBBHSU2016,JungKJBBD18}) and
significant projects (e.g.,~\cite{SewellMyreenKlein2013}).
% The effectiveness of program logics has been demonstrated in
% significant verification projects (e.g.,
% CompCert~\cite{CaoBeringerGruetterDoddsAppel2018}).

A price to pay for enjoying the power of program logics, however, is
the considerable amount of effort often needed in establishing their
soundness and completeness wrt. the baseline semantics -- often an
operational semantics.
There have been a plethora of programming languages designed and
implemented to meet the needs of different domains. 
The recent development of blockchain technology alone has led to the
creation of multiple languages, such as Solidity~\cite{Solidity},
Yul~\cite{Yul}, Scilla~\cite{SergeyNagarajJohannsenTrunovHao2019},
Move~\cite{Move}, Michelson~\cite{Michelson}, EVM bytecode
language~\cite{Wood2017}, etc.
Developing one program logic for each language that could be used in
scenarios where correctness is of serious concern would require a
huge amount of efforts.

To combat the cumbersomeness of direct program proofs based on
operational semantics, while avoiding the full complexity in the
development of program logics,
% of developing one program logic for each specific
% language,
one could seek to establish the infrastructure necessary for reasoning
about specific kinds of language features, for any languages with
those features.
The results in \cite{Moore2003} and \cite{MoorePenaRosu2018} show how
to deal with fundamental language features that may cause unbounded
behavior, such as iteration and recursion, in a language-independent
fashion.
In \cite{Moore2003}, a technique is proposed to generate inductive 
invariants from annotated loop invariants.
% , in verifying programs whose execution can be captured by a
% small-step relation. 
%
In \cite{MoorePenaRosu2018}, a method is presented to turn the
semantics of a programming language into a program verifier by
applying coinductive reasoning principles.
% 
% The development is again built on a small-step execution relation.
Both developments are built on the small-step execution relation of a
generic programming language.

Small-step semantics~\cite{Plotkin1981} is known to be a fine-grained
approach to the definition of operational semantics.
%programming language semantics.
%
It supports a way to model concurrent 
execution.
It also enables the differentiation of looping and abnormal
termination.
%in the execution of programs.
%
Big-step semantics (or natural
semantics~\cite{ClementDespeyrouxDespeyrouxKahn86,Kahn1987}), on the
other hand, can be easier to formulate.
For instance, the design of the semantic configurations need not track
the intermediate control states. 
% , e.g., with stacks and other delicate structures.
%    
Big-step semantics can also be easier to use.
It does not require the consideration of both derivation sequences and
derivation trees at the same time, in performing proofs.
There exist many formalizations of big-step
semantics (e.g.,~\cite{NipkowOheimb1998,KleinNipkow2005,Schirmer2008,BlazyLeroy2009,Hirai2017,YangLei2018})
with practical uses.

In this work, we propose a technique for reasoning about iteration and
recursion in deductive program verification based on big-step
operational semantics.
For any programming language with a big-step semantics, once a generic
predicate is defined to hold on the premises and corresponding
conclusions for the semantic rules, a theorem becomes available -- the
theorem turns the verification of partial correctness results into
symbolic execution of the target program with auxiliary information
from the user specification.
For loops and recursive function calls, this auxiliary information is
provided in the same form via the specification, enabling the same
pattern of reasoning.
% In addition, the technique enables the loops and recursive function
% calls contained in a program to be treated uniformly -- with one
% single specification style and with the same pattern of reasoning.
%
We illustrate our technique using verification tasks involving
simplified imperative and functional languages.
We  mechanize the proofs of all formal results~\cite{Li2021} in
the Coq proof assistant~\cite{Coq}.

The main technical contributions of this article are:

\begin{itemize}
\item a language-independent technique simplifying the deductive
  verification of iterative and recursive program structures based on
  big-step semantics,
\item proofs for the soundness and relative completeness of the
  technique,
\item illustration of the technique with the verification of example
  programs in simplified programming languages of different paradigms,
\item mechanization of proofs and verification examples in the Coq
  proof assistant.
\end{itemize}

% Through the Coq development,
We provide an infrastructure that handles the routine part of the work
in reasoning about programming constructs with potentially unbounded
behavior, based on a common model of big-step execution in a proof
assistant.
This provides a basis for a language-independent deductive program
verifier.

%
% This infrastructure can be used as a starting point for building a
% full-fledged language-independent program verifier in Coq.

% In a formal sense, we capture the essence in the reasoning about
% programming constructs giving rise to repetitive behavior based on
% invariants.
% %
% We also provide an infrastructure that handles the routine part of the
% work in reasoning about iteration and recursion in verifying programs
% using a proof assistant. 

\paragraph{Structure.} The remaining part of this article is structured 
as follows.
In \cref{sec:related-work}, we discuss related work. 
% existing work related to ours.
%
In \cref{sec:method}, we introduce the reasoning technique, and prove
its soundness.
In \cref{sec:illustrative-example}, we illustrate the technique with a
toy example that is developed in detail.
In \cref{sec:verification}, we present further verification examples
targeting simplified imperative and functional languages.
In \cref{sec:completeness}, we discuss the completeness of the
verification technique. %, giving formal results.
In \cref{sec:discussion}, we discuss potential
% further applications and
improvements of our technique. 
Finally, we conclude in \cref{sec:conclusion}.

\vspace*{-1.5ex}
\section{Related Work}\label{sec:related-work}
\vspace*{-1ex}

Inductive invariants~\cite{McCarthy1962} are well-studied means to
sound program verification directly based on operational execution
models.
To avoid the difficulty of specifying an inductive invariant that must
be preserved by all the atomic steps that can be performed by a
program, a method is proposed to generate inductive invariants from
inductive assertions~\cite{Moore2003}.
The verification of the generated inductive invariants concludes the
verification of the target program.
The soundness result of the generation is proven once and can be used
for different languages with a small-step execution relation.
In comparison to this work, our technique targets big-step
operationals semantics, and its soundness does not rely on the
reduction of the verification problem to the generation of inductive
invariants.
In \cite{MoorePenaRosu2018}, a technique is proposed to generate sound
program verifiers based on existing formalizations of small-step
semantics in proof assistants.
The soundness of the technique is established with an coinductive
argument.
In comparison, our technique targets big-step operational semantics,
and is based on inductive reasoning.
Nevertheless, we are inspired by this work in the style of
language-independent program specifications and the form of
completeness statements.

In \cite{StefanescuParkYuwenLiRosu2016}, a language-independent
verification technique based on reachability logics and semantics
formulated in rewriting systems is introduced.
In comparison, our technique can only be used for big-step semantics.
However, our technique can be used with semantic definitions using
inductive predicates in a proof assistant, and requires only the
logical foundation of the proof assistant to function.
Our technique also has a succinct, inductive argument for soundness.  

% The language-independent verification technique in
% \cite{StefanescuParkYuwenLiRosu2016} is based on reachability logic,
% and is applicable to different types of semantics specified in a
% term-rewriting system.
% %
% In comparison, our technique can only be used for big-step semantics.
% %
% On the other hand, it can be used with semantic definitions using
% inductive predicates in a proof assistant, and requires only the
% logical foundation of the proof assistant to function.

Several developments provide means to systematically derive abstract
semantics from concrete semantics such as big-step operational
semantics and its
variants~\cite{BodinGardnerJensenSchmitt2019,Schmidt1995,BodinJensenSchmitt2015}.
Among these, \cite{BodinGardnerJensenSchmitt2019} proposes a
language-independent notion of skeletal semantics that can be
instantiated to obtain concrete and abstract semantic interpretations.
However, the emphasis of these developments is in obtaining automated
static analyses of programs, rather than in exploiting user-provided
specification in the deductive verification of deep correctness
properties.

To some extent, language-independent program verification can also be
supported by encoding the target languages or target programs in the
same language (e.g., WhyML, Boogie, etc.) or calculus (e.g.,
CSP, the $\pi$-calculus, etc.) supporting verification.
%(e.g., in \cite{GreenawayAndronickKlein2012}).
%
This encoding can be considerably more light-weight than the direct
formalization of the syntax and semantics of the source language.
However, when the features of the source language are sufficiently
complicated, it can be highly non-trivial to justify the encoding.

\vspace*{-1.5ex}
\section{The Technique}\label{sec:method}
\vspace*{-1ex}

%The verification technique presented in this section

Our verification technique can be used to check that the potential
execution results of a program satisfy pre-specified conditions.
The potential execution results are estimated by a combination of
concrete computation according to the big-step semantics of the
programming language, and abstract inference according to the
auxiliary information in the specification.
The abstract inference helps realized what is usually accomplished wit
loop invariants in reasoning about loops, and with function contracts
in reasoning about function calls.
% The abstract inference can be leveraged to realize what is usually
% accomplished with loop invariants in reasoning about loops, and with
% function contracts
% % (with pre-conditions and post-conditions)
% in reasoning about function calls.

\vspace*{-2ex}
\subsection{Specifications}
\vspace*{-.5ex}
Let $\configs$ be the set of \emph{configurations} ranged over by
$\configsMeta$.
Let $\rConfigs$ be the set of \emph{result configurations} ranged over by
$\rConfigsMeta$.
A concrete example for a configuration would be a pair of a program
and a state (which may have its own structure).
A concrete example for a result configuration would be a state.

A specification is a function $\specsMeta \in \configs\to \powerSet{\rConfigs}$.
For a configuration $\configsMeta$, if $\configsMeta$ contains the
complete program to be verified, then $\specsMeta(\configsMeta)$ is
the set capturing the required range for the results of executing the
program from $\configsMeta$.
Otherwise, $\specsMeta(\configsMeta)$ is the expected set of potential
results obtained by executing some statement within the overall
program.
This set provides auxiliary information for the verification.

% We consider \emph{specifications} that are functions mapping each
% configuration to a set of result configurations.
% %
% For a specification $\specsMeta\in \configs\to \powerSet{\rConfigs}$
% and a configuration $\configsMeta$,
% %
% each potential result of execution starting from $\configsMeta$ is
% expected to be in the set $\specsMeta(\configsMeta)$.

% A specification has two main types of functionality. 
% %
% The first is to describe the potential results of executing the
% overall target program.
% %
% The second is to hold auxiliary information about the potential
% execution results of selected statements in the target program, so as
% to support the verification of the target program.

% The verification approach to be presented in this section checks that
% the potential execution results of a program are among the candidates
% specified by the user.
% %
% In characterizing the possible execution results, the approach
% combines the concrete computation of the results according to the
% big-step semantics of the programming language, and the abstract
% inference of the results according to the specification by the user.
% %
% The abstract inference can be leveraged to realize what is usually
% accomplished with loop invariants in reasoning about loops, and with
% function contracts (with pre-conditions and post-conditions) in
% reasoning about function calls.

\vspace*{-2ex}
\subsection{Semantic Derivation and Correctness}\label{ssec:semantic-derivation}
\vspace*{-.5ex}

We model the set of rules of a big-step operational semantics by a
predicate
$\mi{rule}\in (\configs\times\rConfigs)^*\to
(\configs\times\rConfigs)\to\{\TT,\FF\}$.
Each semantic rule is captured as 

\vspace*{-2ex}
\small
\[
  \RULE{\concl{\configsMeta}{\rConfigsMeta}}
  {\premise{\configsMeta_1}{\rConfigsMeta_1},\dots,
  \premise{\configsMeta_n}{\rConfigsMeta_n}} 
\]
\normalsize
Here, the list
$\LST{\premise{\configsMeta_1}{\rConfigsMeta_1}, \dots,
  \premise{\configsMeta_n}{\rConfigsMeta_n}}$ models the list of
premises of the rule, and $\concl{\configsMeta}{\rConfigsMeta}$ models
the conclusion of the rule.
Each premise or conclusion consists of a configuration in the set
$\configs$ and a corresponding result configuration in the set
$\rConfigs$.
% Here, $\premise{\configsMeta_1}{\rConfigsMeta_1}$, \dots,
% $\premise{\configsMeta_n}{\rConfigsMeta_n}$ are the premises,
% and $\concl{\configsMeta}{\rConfigsMeta}$ is the conclusion. 
%
A side condition in a semantic rule can be captured by a condition on
the parameters $\configsMeta_1$, \dots, $\configsMeta_n$,
$\rConfigsMeta_1$, \dots, $\rConfigsMeta_n$, $\configsMeta$, and
$\rConfigsMeta$, in the concrete definition of $\mi{rule}$.

A semantic derivation concluding that the configuration $\configsMeta$
can be evaluated to the result configuration $\rConfigsMeta$ in the
big-step semantics is captured by

\vspace*{-2ex}
\small
\begin{align*}
  \DERIV{\configsMeta,\rConfigsMeta}
  ~\defEq~
  \exists k\!:
  &\,
    \exists
    \configsMeta_1,\dots,\configsMeta_k\!:
    \exists
    \rConfigsMeta_1,\dots,\rConfigsMeta_k\!: \\
  &\,
    \RULE{\concl{\configsMeta}{\rConfigsMeta}}
    {\premise{\configsMeta_1}{\rConfigsMeta_1},\dots,
    \premise{\configsMeta_k}{\rConfigsMeta_k}}
    \, \land \, 
    \forall i\in \{1,\dots,k\}\!:
    \DERIV{\configsMeta_i,\rConfigsMeta_i} 
\end{align*}
\normalsize
% \begin{align*}
%   \DERIV{\configsMeta,\rConfigsMeta}
%   ~\defEq~
%   \exists
%   \configsMeta_1,\dots,\configsMeta_k:
%   \exists
%   \rConfigsMeta_1,\dots,\rConfigsMeta_k:
%   &\,
%   \RULE{\concl{\configsMeta}{\rConfigsMeta}}
%   {\premise{\configsMeta_1}{\rConfigsMeta_1},\dots,
%     \premise{\configsMeta_k}{\rConfigsMeta_k}}
%     \, \land \\
%   &\,\,
%   \forall i\in \{1,\dots,k\}:
%   \DERIV{\configsMeta_i,\rConfigsMeta_i} 
% \end{align*}
%
Hence, the configuration $\configsMeta$ can be evaluated to the result
configuration $\rConfigsMeta$, or $(\configsMeta,\rConfigsMeta)$ can
be derived in the big-step semantics, if there is a semantic rule with
$(\configsMeta,\rConfigsMeta)$ as conclusion, and each premise of the
rule can itself be derived in the big-step semantics.
Intuitively, if $\DERIV{\configsMeta,\rConfigsMeta}$ can be
established, then there is a finite derivation tree rooted at
$\concl{\configsMeta}{\rConfigsMeta}$.

With the notion of semantic derivation defined above, we formalize the
notion of partial correctness as the validity of specifications.
%
% We express the validity of specifications $\specsMeta$ using the
% predicate $\mi{valid}$.
%

\vspace*{-2ex}
\small
\[
  % \claimValid{\configsMeta}{\setsOfRConfigsMeta}
  % ~\defEq~&
  %           \forall \rConfigsMeta: \DERIV{\configsMeta,\rConfigsMeta}
  %           \Rightarrow \rConfigsMeta\in\setsOfRConfigsMeta 
  % \\
  \specValid{\specsMeta}
  ~\defEq~
  \forall \configsMeta, \rConfigsMeta:
  \DERIV{\configsMeta, \rConfigsMeta} 
  \Rightarrow
  \rConfigsMeta\in \specsMeta(\configsMeta) 
\]
\normalsize
A specification $\specsMeta$ is valid, if for each configuration
$\configsMeta$, any result configurations semantically derivable from
$\configsMeta$ is a member of $\specsMeta(\configsMeta)$.

\vspace*{-2ex}
\subsection{Specification-aware Inference and Verification}
\vspace*{-.5ex}

We infer the potential execution results of a configuration under a
given specification $\specsMeta$ according to the following
definition.

\vspace*{-2ex}
\small
\begin{align*}
  \INFER{\configsMeta}{\rConfigsMeta}{\specsMeta}
  ~\defEq~
  &
    \exists k\!:    
    \configsMeta_1,\dots,\configsMeta_k\!:
    \exists \rConfigsMeta_1,\dots,\rConfigsMeta_k\!:\,\\[-.5ex]
  &\quad~~\,
    \RULE{\!\concl{\configsMeta}{\rConfigsMeta}}
    {\premise{\configsMeta_1}{\rConfigsMeta_1},\dots,
    \premise{\configsMeta_k}{\rConfigsMeta_k}} 
    \land \,
    \forall i\!\in\! \{1,\dots,k\}: 
    \rConfigFor{\configsMeta_i}{\rConfigsMeta_i}{\specsMeta}
  \\
  \rConfigFor{\configsMeta}{\rConfigsMeta}{\specsMeta}
  ~\defEq~
  &
    \rConfigsMeta\in \specsMeta(\configsMeta) \,\land\,
    (\specsMeta(\configsMeta)=\rConfigs
    \Rightarrow
    \INFER{\configsMeta}{\rConfigsMeta}{\specsMeta}
    )
    % \rConfigFor{\configsMeta}{\rConfigsMeta}{\specsMeta}
    % ~\defEq~
    % &
    % (\specsMeta(\configsMeta) \neq \bot
    % \land \rConfigsMeta\in \specsMeta(\configsMeta)
    % \,\lor\,
    % \specsMeta(\configsMeta) = \bot
    % \land
    % \INFER{\configsMeta}{\rConfigsMeta}{\specsMeta}
    % )
\end{align*}
\normalsize
% \begin{align*}
%   \INFER{\configsMeta}{\rConfigsMeta}{\specsMeta}
%   ~\defEq~
%   &
%   \exists
%   \configsMeta_1,\dots,\configsMeta_k:
%   \exists \rConfigsMeta_1,\dots,\rConfigsMeta_k:\,
%     \RULE{\concl{\configsMeta}{\rConfigsMeta}}
%           {\premise{\configsMeta_1}{\rConfigsMeta_1},\dots,
%             \premise{\configsMeta_k}{\rConfigsMeta_k}} 
%     \,\land  \\[-.5ex]
%   &\,\qquad\qquad\qquad\qquad\qquad\quad
%     \forall i\in \{1,\dots,k\}: 
%     \rConfigFor{\configsMeta_i}{\rConfigsMeta_i}{\specsMeta}
%   \\
%      \rConfigFor{\configsMeta}{\rConfigsMeta}{\specsMeta}
%   ~\defEq~
%   &
%   (\specsMeta(\configsMeta) \neq \bot 
%     \land \rConfigsMeta\in \specsMeta(\configsMeta) 
%     \,\lor\,
%   \specsMeta(\configsMeta) = \bot
%     \land
%     \INFER{\configsMeta}{\rConfigsMeta}{\specsMeta}
%     )   
% \end{align*}
%
The result configuration $\rConfigsMeta$ is infered from the
configuration $\configsMeta$ with the help of the specification
$\specsMeta$, if there is a semantic rule with
$(\configsMeta,\rConfigsMeta)$ as conclusion, and for each premise
$(\configsMeta_i,\rConfigsMeta_i)$ of the semantic rule,
$\rConfigsMeta_i$ is a potential result for $\configsMeta_i$ according
to $\specsMeta$, as is captured by the auxiliary predicate
$\mi{res}^{\specsMeta}$.
The expression
$\rConfigFor{\configsMeta_i}{\rConfigsMeta_i}{\specsMeta}$ says that
the possible candidates for $\rConfigsMeta_i$ are constrained by the
information contained in the specification about $\configsMeta_i$.
In addition, if $\specsMeta$ does not provide any useful information
about $\configsMeta_i$ (i.e., $\specsMeta(\configsMeta_i)=\rConfigs$),
then $\rConfigsMeta_i$ should be inferable from $\configsMeta_i$.

Intuitively, the application of the semantic rules in the inference 
corresponds to the symbolic execution of the target program. 
The information in the specification can be used to overcome 
the inability to symbolically execute the constructs with potentially
unbounded behavior, such as iteration and recursion.

We formulate the 
% correctness 
condition to be verified on
specifications $\specsMeta$ using the predicate $\mi{verif}$.
In other words, $\verifySpec{\specsMeta}$ is the syntactical
correctness condition. 

\vspace*{-2.5ex}
\small
\begin{align*}
  \verifySpec{\specsMeta}
  ~\defEq~&
            \forall \configsMeta, \rConfigsMeta:
            \INFER{\configsMeta}{\rConfigsMeta}{\specsMeta} 
            \Rightarrow \rConfigsMeta\in \specsMeta(\configsMeta)
            % \checkClaim{\configsMeta}{\setsOfRConfigsMeta}
            % {\specsMeta}
            % ~\defEq~&
            % \forall \rConfigsMeta:
            % \INFER{\configsMeta}{\rConfigsMeta}{\specsMeta}
            % \,\Rightarrow\, \rConfigsMeta\in \setsOfRConfigsMeta \\
\end{align*}
\normalsize
A specification $\specsMeta$ is verified, if for each configuration
$\configsMeta$, any result configurations that can be infered from
$\configsMeta$ with the help of $\specsMeta$ are contained in
$\specsMeta(\configsMeta)$.

\vspace*{-2ex}
\subsection{Soundness} \label{ssec:soundness}
\vspace*{-.5ex}

% We say that a \emph{claim} $\CLAIM{\configsMeta}{\setsOfRConfigsMeta}$
% \emph{is valid} if for all $\rConfigsMeta$ such that
% $\DERIV{\configsMeta,\rConfigsMeta}$, it holds that
% $\rConfigsMeta\in\setsOfRConfigsMeta$.
% %
% We say that a \emph{specification} $\specsMeta$ \emph{is valid} if 
% all the claims in $\specsMeta$ are valid.

%
% A claim $\CLAIM{\configsMeta}{\setsOfRConfigsMeta}$ is valid,
% if all result configurations for $\configsMeta$ are in
% $\setsOfRConfigsMeta$.
% %
% A specification $\specsMeta$ is valid, if all claims in $\specsMeta$
% are valid.

%
% We present the results on the soundness and completeness of the
% verification techniuqe.
% %
% A reader who is interested more in the application of the technique
% than in its theoretical underpinings may skip the proofs on first
% reading.

We prove the implication from $\verifySpec{\specsMeta}$ to 
$\INFER{\configsMeta}{\rConfigsMeta}{\specsMeta}$.
The following lemma is a key component of this proof. 
% Under the condition of $\verifySpec{\specsMeta}$, it can be shown that
%  over-approximates
% $\DERIV{\configsMeta, \rConfigsMeta}$ in the set of result
% configurations charaterized. 
%
\begin{lemma}\label{lem:sound}
  If $\verifySpec{\specsMeta}$, and 
  $\DERIV{\configsMeta, \rConfigsMeta}$ holds,
  then $\INFER{\configsMeta}{\rConfigsMeta}{\specsMeta}$ holds. 
\end{lemma}

\begin{proof}
  According to the definition of
  $\DERIV{\configsMeta, \rConfigsMeta}$, if this predicate holds, then
  there is a finite derivation tree generated by the following
  inference rule.

  \vspace*{-2ex}
  \small
  \[
    \Inference
    {
      \DERIV{\configsMeta_1,\rConfigsMeta_1}
      \quad \dots \quad 
      \DERIV{\configsMeta_m,\rConfigsMeta_m}
      \qquad\qquad
      \RULE
      {\concl{\configsMeta}{\rConfigsMeta}} 
      {
        \premise{\configsMeta_1}{\rConfigsMeta_1}, \dots,
        \premise{\configsMeta_m}{\rConfigsMeta_m} 
      }
    }
    {
      \DERIV{\configsMeta,\rConfigsMeta}
    }
  \]
  \normalsize
  The proof is by induction on the derivation tree for
  $\DERIV{\configsMeta,\rConfigsMeta}$.
  
  From $\DERIV{\configsMeta, \rConfigsMeta}$,
  we have $\DERIV{\configsMeta_1,\rConfigsMeta_1}$, \dots,
  $\DERIV{\configsMeta_m,\rConfigsMeta_m}$, and

  \vspace*{-2ex}
  \small
  \begin{equation}\label{eq:rule_m}
    \RULE
      {\concl{\configsMeta}{\rConfigsMeta}} 
      {
      \premise{\configsMeta_1}{\rConfigsMeta_1}, \dots,
      \premise{\configsMeta_m}{\rConfigsMeta_m} 
    }
  \end{equation}
  \normalsize
  for some $m$, $\configsMeta_1$, \dots, $\configsMeta_m$,
  $\rConfigsMeta_1$, \dots, $\rConfigsMeta_m$.

  For each $i\in \{1,\dots,m\}$, we have 
  $\INFER{\configsMeta_i}{\rConfigsMeta_i}{\specsMeta}$ 
  from $\DERIV{\configsMeta_i,\rConfigsMeta_i}$ and the 
  induction hypothesis. 
  We show that
  $\rConfigFor{\configsMeta_i}{\rConfigsMeta_i}{\specsMeta}$ holds by
  distinguishing between the cases where
  $\specsMeta(\configsMeta_i)=\rConfigs$ and where
  $\specsMeta(\configsMeta_i)\neq\rConfigs$.
  \begin{itemize}
  \item Suppose $\specsMeta(\configsMeta_i)=\rConfigs$.
    Then, it holds that $\rConfigsMeta\in \specsMeta(\configsMeta_i)$.
    Hence, we have $\rConfigFor{\configsMeta_i}{\rConfigsMeta_i}{\specsMeta}$
    because of $\INFER{\configsMeta_i}{\rConfigsMeta_i}{\specsMeta}$,
    and the definition of $\mi{res}^\specsMeta$. 
    \smallskip\smallskip
  \item Suppose $\specsMeta(\configsMeta_i)\neq\rConfigs$.
    From $\INFER{\configsMeta_i}{\rConfigsMeta_i}{\specsMeta}$, and
    $\verifySpec{\specsMeta}$, we have
    $\rConfigsMeta_i\in \specsMeta(\configsMeta_i)$.
    Hence, we have
    $\rConfigFor{\configsMeta_i}{\rConfigsMeta_i}{\specsMeta}$
    according to the definition of $\mi{res}^{\specsMeta}$. 
  \end{itemize}
  Hence, for each $i\in \{1,\dots,m\}$, we have
  $\rConfigFor{\configsMeta_i}{\rConfigsMeta_i}{\specsMeta}$.
  Thus, we can deduce
  $\INFER{\configsMeta}{\rConfigsMeta}{\specsMeta}$ using  
  \cref{eq:rule_m} and the definition of $\mi{infer}^\specsMeta$. 
  This completes the proof. 
  \qed 
\end{proof}

Using this lemma, the soundness theorem can be obtained directly.

\vspace*{-.5ex}
\begin{theorem}[Soundness]\label{thm:soundness}
  If $\verifySpec{\specsMeta}$ can be established, then
  $\specValid{\specsMeta}$ holds. 
\end{theorem}
\vspace*{-3ex}
\begin{proof}
  Assume $\verifySpec{\specsMeta}$ and
  $\DERIV{\configsMeta,\rConfigsMeta}$.
  Then, we have $\INFER{\configsMeta}{\rConfigsMeta}{\specsMeta}$
  according to \cref{lem:sound}.
  Thus, we can deduce $\rConfigsMeta\in\specsMeta(\configsMeta)$ using
  $\verifySpec{\specsMeta}$.
  \qed
\end{proof}
The application of this theorem reliably turns the problem of
establishing the validity of a specification $\specsMeta$ into the
problem of proving $\verifySpec{\specsMeta}$, 
irrespective of the language used for the program that is specified in
$\specsMeta$.

% \begin{proof}
%   Assume $\verifySpec{\specsMeta}$,
%   $\specsMeta(\configsMeta)=\setsOfRConfigsMeta$, 
%   and $\DERIV{\configsMeta,\rConfigsMeta}$.
%   %
%   We then have $\rConfigsMeta\in\setsOfRConfigsMeta$
%   using \cref{lem:sound}. 
%   %
%   This completes the proof of the theorem. 
%   \qed 
% \end{proof}

\vspace*{-1.5ex}
\begin{remark}
  \cref{lem:sound} suggests that an abstract form of computation is
  obtained leveraging user specification that is verified.
  In abstract interpretation~\cite{CousotCousot1977}, the focus is
  often to calculate the abstract form of the computation performed by
  each kind of program statement, to support automated program
  analysis.
  On the other hand, our focus is to leverage user specification for
  specific concrete statements in a program, to provide hints in a
  deductive program verification. 
  %
  % \cref{lem:sound} suggests that an abstract variant of computation is
  % obtained leveraging user specification that is verified.
  % %
  % Different from the focus in abstract
  % interpretation~\cite{CousotCousot1977} that is the calculation of
  % the abstract variant of each kind of program statements, our focus
  % is to enable the use of auxiliary information contained in the user
  % specification for specific statements in a program. 
\end{remark}

\vspace*{-3ex}
\section{Illustrative Example}\label{sec:illustrative-example}
\vspace*{-1ex}

In this section, we illustrate our technique using a toy example. 
In this example, a program computing the factorial of a natural number is written in the While language~\cite{NielsonNielson2007}. 
We show how the big-step semantics of the While language can be
formulated with the $\mi{rule}$ predicate introduced in
\cref{ssec:semantic-derivation}.
We then show how the functional correctness of the factorial program can be specified and proven.

\vspace*{-2ex}
\subsection{Big-step Semantics of the While Language}
\vspace*{-1ex}

The While language consists of arithmetic expressions $a$,
Boolean expressions $b$, and statements $\stmtsMeta\in \stmts$.
A statement can be $\SKIP$ that performs no operation, an assignment
$\ASSG{\varsMeta}{a}$, a sequential composition
$\stmtsMeta_1;\stmtsMeta_2$, a branching statement
$\IF{b}{\stmtsMeta_1}{\stmtsMeta_2}$, or a loop
$\WHILE{b}{\stmtsMeta}$.

For programs in the While language, the states $\whStatesMeta$ are
elements of $\whStates\defEq \vars\to\ints$.
Here, $\vars$ is the set of variables and $\ints$ is the set of
integers.
The evaluation of arithmetic expressions and Boolean expressions in
states can be formalized by defining the evaluation functions
$\mathcal{A}$ and $\mathcal{B}$, respectively, as in
\cite{NielsonNielson2007}. 
The set $\configs$ of configurations is $\stmts\times \whStates$.
The set $\rConfigs$ of result configurations is $\whStates$.
We formulate the big-step semantics by defining the predicate
$\mi{rule}$, as in \cref{fig:sem-rules-while-language}.
In each line, a combination of the parameter values for which
$\mi{rule}$ holds is given.

%
% We assume that $\vars$ is the set of program variables,
% %
% $\ints$ is the set of integers, 
% %
% $\whStatesMeta$ is a state in $\vars\to\ints$,
% %
% $a$ is an arithmetic expression,
% %
% $b$ is a Boolean expression, 
% %
% and $\mathcal{A}$ and $\mathcal{B}$ are the evaluation
% functions~\cite{NielsonNielson2007} for arithmetic expressions and
% Boolean expressions, respectively.
% %
% We formulate the big-step semantics of this language by defining the
% predicate $\mi{rule}$, as in \cref{fig:sem-rules-while-language}.
% %
% In each line, a combination of the parameters for which $\mi{rule}$
% holds is given.

\begin{figure}[h]
\small
\[
  \begin{array}{l}
    {~}\\[-7ex]
    \RULE{\concl{\whConfig{\SKIP}{\whStatesMeta}}{\whStatesMeta}}{\,}
    \\\\[-2ex]
    \RULE{\concl{\whConfig{\ASSG{x}{a}}{\whStatesMeta}}
    {\whStatesMeta[x\mapsto \whAEval{a}\whStatesMeta]}}{\,}
    \\\\[-2ex]
    \RULE{\concl{\whConfig{\SEQ{\stmtsMeta_1}{\stmtsMeta_2}}{\whStatesMeta}}
    {\whStatesMeta'}}
    {\premise{\whConfig{\stmtsMeta_1}{\whStatesMeta}}{\whStatesMeta''},
    \premise{\whConfig{\stmtsMeta_2}{\whStatesMeta''}}{\whStatesMeta'}}
    \\\\[-2ex]
    \RULE{\concl{\whConfig{\IF{b}{\stmtsMeta_1}{\stmtsMeta_2}}{\whStatesMeta}}
    {\whStatesMeta'}}
    {\premise{\whConfig{\stmtsMeta_1}{\whStatesMeta}}{\whStatesMeta'}}
    \quad \mrm{if}~\whBEval{b}\whStatesMeta=\TT
    \\\\[-2ex]
    \RULE{\concl{\whConfig{\IF{b}{\stmtsMeta_1}{\stmtsMeta_2}}{\whStatesMeta}}
    {\whStatesMeta'}}
    {\premise{\whConfig{\stmtsMeta_2}{\whStatesMeta}}{\whStatesMeta'}}
    \quad \mrm{if}~\whBEval{b}\whStatesMeta=\FF
    \\\\[-2ex]
    \RULE{\concl{\whConfig{\WHILE{b}{\stmtsMeta}}{\whStatesMeta}}
    {\whStatesMeta'}}
    {\premise{\whConfig{\stmtsMeta}{\whStatesMeta}}{\whStatesMeta''},
    \premise{\whConfig{\WHILE{b}{\stmtsMeta}}{\whStatesMeta''}}{\whStatesMeta'}}
    ~\mrm{if}~\whBEval{b}\whStatesMeta=\TT
    \\\\[-2ex]
    \RULE
    {\concl{\whConfig{\WHILE{b}{\stmtsMeta}}{\whStatesMeta}}{\whStatesMeta}}{\,}
    ~\mrm{if}~\whBEval{b}\whStatesMeta=\FF
  \end{array}
\]
\vspace*{-2ex}
\caption{The semantic rules for the statements of the While language}
\label{fig:sem-rules-while-language}
\vspace*{-3ex}
\normalsize
\end{figure}

\vspace*{-1.5ex}
\subsection{Factorial Program and its Specification}
\vspace*{-.5ex}

Consider the program $\stmtFac$ in the While language.
The program computes the factorial $\FACT{m}$ where $m$ is the initial
value of the program variable $\pgVar{m}$.

\vspace*{-3ex}
\small
\begin{align*}
  \stmtFac~\defEq~
  & (\ASSG{\pgVar{fac}}{\pgVar{m}}; \stmtWh) \\
  \stmtWh ~\defEq~
  &
    (\mathsf{while}~1<\pgVar{m}~\mathsf{do}~
    (\ASSG{\pgVar{m}}{\pgVar{m}-\pgVar{1}}; 
    \ASSG{\pgVar{fac}}{\pgVar{fac}*\pgVar{m}}))
%    \mathsf{od})
\end{align*}
\normalsize
\vspace*{-3ex}

Let $\Pm{m}$ be the set of states where $\pgVar{fac}$ has the value $\FACT{m}$.
Let $\Pmfac{m}{\mi{fac}}$ be the set of states where $\pgVar{fac}$ has
the value $\mi{fac}*\FACT{(m-1)}$. 

\vspace*{-3ex}
\small
\begin{align*}
  \Pm{m}~\defEq~
  &
  \{\whStatesMeta'[\pgVar{fac}\mapsto \FACT{m}]\mid
    \whStatesMeta'\in\whStates\}  \\
  \Pmfac{m}{\mi{fac}}~\defEq~
  &
  \{
  \whStatesMeta'[\pgVar{fac}\mapsto \mi{fac}*\FACT{(m-1)}]
  \mid
  \whStatesMeta'\in\whStates 
  \}
\end{align*}
\normalsize
\vspace*{-3ex}

We consider the following specification for the program.

\vspace*{-3ex}
\small
\[
  \begin{array}{rll}
    \specFac(\whConfig{\stmtFac}{\whStatesMeta})
    ~\defEq~
    &
      \Pm{m} 
      \quad
    &
      \mrm{if}~
      m = \whStatesMeta(\pgVar{m})\,\land\,
      m > 0\,\land\,
      \whStatesMeta \in \whStates
    \\\\[-2ex]
    \specFac(\whConfig{\stmtWh}{\whStatesMeta}) 
    ~\defEq~
    &
      \Pmfac{m}{\mi{fac}} 
      \quad
    &
      \mrm{if}~
      m = \whStatesMeta(\pgVar{m})\,\land\,
      m > 0\,\land\, 
      \mi{fac} = \whStatesMeta(\pgVar{fac})\,\land\,
      \whStatesMeta \in \whStates
    \\\\[-2ex]
    \specFac(\configsMeta)
    ~\defEq~
    &
      \whStates
     % \rConfigs
      \quad
    &
      \mrm{if}~\configsMeta~\mrm{is~not~of~the~above~forms}
  \end{array}
\]
\normalsize
The specification says that when $\stmtFac$ finishes execution started
in a state where the value of $\pgVar{m}$ is $m>0$, the value of
$\pgVar{fac}$ will be $\FACT{m}$.
% where $m$ is the value of $\pgVar{m}$ in the initial state for $\stmtFac$.
%
The specification also contains the auxiliary claim that when the loop
$\stmtWh$ finishes execution started in a state where $\pgVar{fac}$
has the value $\mi{fac}$ and $\pgVar{m}$ has the value $m>0$, the
value of $\pgVar{fac}$ will be equal to the product of $\mi{fac}$ and
$\FACT{(m-1)}$ (noting that $\FACT{0}=1$).

%  where $\mi{fac}$ and $m$
% are the values of $\pgVar{fac}$ and $\pgVar{m}$ in the initial state
% for the loop $\stmtWh$.

\vspace*{-1.5ex}
\subsection{Proof of the Factorial Program}
\vspace*{-.5ex}

A direct proof of the factorial program $\stmtFac$ based on the
big-step operational semantics of the While language would require an
induction on the shape of derivation trees
(e.g.,~\cite{NielsonNielson2007}) to establish a suitable invariant
for the loop $\stmtWh$.

Using the technique of \cref{sec:method}, we aim at establishing
$\specValid{\specFac}$.
With \cref{thm:soundness}, it suffices to show
$\verifySpec{\specFac}$ -- 
for all $\configsMeta$ and $\rConfigsMeta$, assuming
$\INFER{\configsMeta}{\rConfigsMeta}{\specFac}$, we attempt to show
$\rConfigsMeta\in\specFac(\configsMeta)$.

\begin{enumerate} 
\item Firstly, assume $\configsMeta$ is
  $\whConfig{\stmtFac}{\whStatesMeta}$, where 
  $\pgVarVal{\whStatesMeta}{\pgVar{m}}>0$.
  Then, $\specFac(\configsMeta)$ is $\Pm{m}$, where
  $m=\pgVarVal{\whStatesMeta}{\pgVar{m}}$.
  Using
  $\INFER{\whConfig{\stmtFac}{\whStatesMeta}}
  {\rConfigsMeta}{\specFac}$ and the semantics of the While language
  in \cref{fig:sem-rules-while-language}, it is not difficult to
  obtain
  
  \vspace*{-2ex}
  \small
  \[
    \RULE {(\whConfig{\stmtFac}{\whStatesMeta},\rConfigsMeta)}
    {(\whConfig{\ASSG{\pgVar{fac}}{\pgVar{m}}}{\whStatesMeta},\whStatesMeta''),
      (\whConfig{\stmtWh}{\whStatesMeta''},\rConfigsMeta)}
  \]
  \normalsize
  for some $\whStatesMeta''$ such that
  $\rConfigFor
  {\whConfig{\ASSG{\pgVar{fac}}{\pgVar{m}}}{\whStatesMeta}}
  {\whStatesMeta''}{\specFac}$ and
  $\rConfigFor{\whConfig{\stmtWh}{\whStatesMeta''}}{\rConfigsMeta}
  {\specFac}$.
  Since
  $\specFac(\whConfig{\ASSG{\pgVar{fac}}{\pgVar{m}}}{\whStatesMeta})=\rConfigs$,
  we deduce 
  $\INFER{\whConfig{\ASSG{\pgVar{fac}}{\pgVar{m}}}{\whStatesMeta}}
  {\whStatesMeta''}{\specFac}$ from
  $\rConfigFor
  {\whConfig{\ASSG{\pgVar{fac}}{\pgVar{m}}}{\whStatesMeta}}
  {\whStatesMeta''}{\specFac}$.
  Hence, we deduce
  $\whStatesMeta''=\whStatesMeta[\pgVar{fac}\mapsto
  \pgVarVal{\whStatesMeta}{\pgVar{m}}]$. 
  Hence, we have
  $\pgVarVal{\whStatesMeta''}{\pgVar{m}}=\pgVarVal{\whStatesMeta}{\pgVar{m}}>0$.
  Hence, we have 
  $\specFac(\whConfig{\stmtWh}{\whStatesMeta''})= \Pmfac
  {\pgVarVal{\whStatesMeta''}{\pgVar{m}}}
  {\pgVarVal{\whStatesMeta''}{\pgVar{fac}}}
  =\{\whStatesMeta'[\pgVar{fac}\mapsto
  \pgVarVal{\whStatesMeta''}{\pgVar{fac}}*
  \FACT{(\pgVarVal{\whStatesMeta''}{m}-1)}]\mid
  \whStatesMeta'\in\whStates\}=
  \{\whStatesMeta'[\pgVar{fac}\mapsto
  \FACT{\pgVarVal{\whStatesMeta}{\pgVar{m}}}] \mid
  \whStatesMeta'\in\whStates\}=
  \setsOfRConfigsMeta_m$.
  Moreover, from 
  $\rConfigFor{\whConfig{\stmtWh}{\whStatesMeta''}}{\rConfigsMeta}
  {\specFac}$
  we have $\rConfigsMeta\in \Pmfac
  {\pgVarVal{\whStatesMeta''}{\pgVar{m}}}
  {\pgVarVal{\whStatesMeta''}{\pgVar{fac}}}$. 
  Ultimately, we have $\rConfigsMeta\in \setsOfRConfigsMeta_{m}$.

  \smallskip\smallskip
\item Secondly, assume $\configsMeta$ is
  $\whConfig{\stmtWh}{\whStatesMeta}$, where
  $\pgVarVal{\whStatesMeta}{\mVar}>0$. 
  Then, $\specFac(\configsMeta)$ is $\Pmfac{m}{\mi{fac}}$, where
  $m=\pgVarVal{\whStatesMeta}{\pgVar{m}}$, and
  $\mi{fac}=\pgVarVal{\whStatesMeta}{\pgVar{fac}}$.
  Using $\INFER{\whConfig{\stmtWh}{\whStatesMeta}}
  {\rConfigsMeta}{\specFac}$ and the semantics of the While language
  in \cref{fig:sem-rules-while-language}, we have the following two cases.
  
  \smallskip\smallskip
  \begin{enumerate}
  \item We have $m\le 1$, 
    $\RULE{(\whConfig{\stmtWh}{\whStatesMeta},\whStatesMeta)}{~}$, and
    $\rConfigsMeta=\whStatesMeta$. 
    Since $m>0$ and $m\le 1$, we have $m=1$.
    Hence, it is not difficult to deduce
    $\rConfigsMeta\in \setsOfRConfigsMeta'_{m,\mi{fac}}$.
    \smallskip\smallskip
  \item We have $m>1$, and 
  
    \vspace*{-2ex}
    \small
    \[
      \RULE{(\whConfig{\stmtWh}{\whStatesMeta},\rConfigsMeta)}
      {(\whConfig{\ASSG{\pgVar{m}}{\pgVar{m}-\pgVar{1}}; 
          \ASSG{\pgVar{fac}}{\pgVar{fac}*\pgVar{m}}}{\whStatesMeta},
        \whStatesMeta''),
        (\whConfig{\stmtWh}{\whStatesMeta''}, \rConfigsMeta)} 
    \]
    \normalsize
    for some $\whStatesMeta''$ such that
    $\rConfigFor{\whConfig{
        \ASSG{\pgVar{m}}{\pgVar{m}-\pgVar{1}}; 
        \ASSG{\pgVar{fac}}{\pgVar{fac}*\pgVar{m}}
      }{\whStatesMeta}}{\whStatesMeta''}
    {\specFac}$ and\\
    $\rConfigFor
    {\whConfig{\stmtWh}{\whStatesMeta''}}{\rConfigsMeta}
    {\specFac}$.
    From the former we have 
    
    \vspace*{-2ex}
    \small
    \[\INFER{\whConfig{
        \ASSG{\pgVar{m}}{\pgVar{m}-\pgVar{1}}; 
        \ASSG{\pgVar{fac}}{\pgVar{fac}*\pgVar{m}}}{\whStatesMeta}}
        {\whStatesMeta''}{\specFac}
    \]
    \normalsize
    The specification $\specFac$ provides no information about the two
    assignments, $\ASSG{\pgVar{m}}{\pgVar{m}-\pgVar{1}}$ and
    $\ASSG{\pgVar{fac}}{\pgVar{fac}*\pgVar{m}}$.
    Hence, $\mi{infer}^{\specFac}$ applies also to these two
    assignments, it can be deduced that
    $\whStatesMeta''=\whStatesMeta [\pgVar{m}\mapsto m-1,
    \pgVar{fac}\mapsto \mi{fac}*(m-1)]$.
    Hence, we have $\pgVarVal{\whStatesMeta''}{\pgVar{m}}=m-1>0$. 
    Hence, $\specFac(\whConfig{\stmtWh}{\whStatesMeta''})=
    \Pmfac
    {\pgVarVal{\whStatesMeta''}{\pgVar{m}}}
    {\pgVarVal{\whStatesMeta''}{\pgVar{fac}}}=
    \{
    \whStatesMeta'[\pgVar{fac}\mapsto
    (\mi{fac}*(m-1))*\FACT{(m-1-1)}] \mid
    \whStatesMeta'\in\whStates 
    \}=
    \Pmfac{m}{\mi{fac}}$.
    Moreover, from
    $\rConfigFor {\whConfig{\stmtWh}{\whStatesMeta''}}{\rConfigsMeta}
    {\specFac}$ we have
    $\rConfigsMeta\in
    \setsOfRConfigsMeta'_{\pgVarVal{\whStatesMeta''}{\pgVar{m}},
      \pgVarVal{\whStatesMeta''}{\pgVar{fac}}}$.
    Ultimately, we have
    $\rConfigsMeta\in \setsOfRConfigsMeta'_{m,\mi{fac}}$.
  \end{enumerate}
\end{enumerate}
In the other cases, we have $\specFac(\configsMeta)=\rConfigs$.
Hence, it trivially holds that
$\rConfigsMeta\in \specFac(\configsMeta)$
The proof is thus complete. 
\qed

The above proof of the factorial program does not require the use of
induction.
Essentially, the induction required for the loop is already
encapsulated in the proof of \cref{thm:soundness}.

\vspace*{-2ex}
\section{Verification of Iterative and Recursive Programs}\label{sec:verification}
\vspace*{-1ex}

In this section, we evaluate our technique with two further examples. 
In the two examples, programming languages of the imperative and
functional paradigms are used, respectively, to implement the
functionality of merging two sorted lists of integers into a single
sorted list of integers.
%
% In \cref{sec:ext-while}, the extension of the While language with
% arrays and functions is used.
% %
% In \cref{ssec:functional}, a fragment of the eager functional language
% discussed in \cite{Reynolds1998} is used.

\vspace*{-1.5ex}
\subsection{Extended While Language and Array-Merging Program}\label{sec:ext-while}
%\vspace*{-1ex}

\subsubsection*{Extended While Language}

The programming language of this section is an extension of the While
language. 
This extension contains the extra features of one-dimensional arrays
and functions.

\vspace*{-2ex}
\paragraph{Syntax.}
We give the syntax for arithmetic expressions $\aExpsMeta$, Boolean
expressions $\bExpsMeta$, and statements $\stmtsMeta$.
We explain the constructs present in the extension only. 
%
% Let $\aExps$ be the set of arithmetic expressions.
% %
% Let $\bExps$ be the set of Boolean expressions.
% %
% Let $\stmts$ be the set of statements.
% %
% Let $\fids$ be the set of function identifiers. 
% %
% The syntax for the entities $\aExpsMeta\in\aExps$,
% $\bExpsMeta\in\bExps$, and $\stmtsMeta\in\stmts$ is given below, where
% $\numeralsMeta\in\numerals$ is a numeral,
% %
% $\varsMeta\in\vars$ is a variable, and
% %
% $\fidsMeta\in\fids$ is a function identifier. 

\vspace*{-2.5ex} 
\small
\begin{align*}
  % \vdsMeta ~\synDefEq~
  % &
  % \varsMeta \mid \varsMeta[\numeralsMeta] \\
  % \pdsMeta ~\synDefEq~
  % &
  % \varsMeta \mid \varsMeta[\,] \\
  \aExpsMeta ~\synDefEq~
    &
      \numeralsMeta \mid \varsMeta \mid
      \arrsMeta \mid \arrRef{\arrsMeta}{\aExpsMeta}
      \mid \addExpr{\aExpsMeta}{\aExpsMeta}
      \mid \subExpr{\aExpsMeta}{\aExpsMeta}
      \mid \mulExpr{\aExpsMeta}{\aExpsMeta}
      \mid \divExpr{\aExpsMeta}{\aExpsMeta}
      % \mid \funCall{\fidsMeta}{\aExpsMeta,\dots,\aExpsMeta}
  \\
  \bExpsMeta ~\synDefEq~
    &
      \trueLit \mid \falseLit \mid 
      \aExpsMeta=\aExpsMeta \mid
      \aExpsMeta < \aExpsMeta \mid 
      \andExpr{\bExpsMeta}{\bExpsMeta} \mid
      % \bExprsMeta || \bExprsMeta \mid
      \notExpr{\bExpsMeta} \\ 
  \stmtsMeta ~\synDefEq~
    &
      \varDecl{\varsMeta} \mid
      \arrDecl{\arrsMeta}{\numeralsMeta} \mid
      \ASSG{\varsMeta}{\aExpsMeta} \mid
      \ASSG{\arrRef{\arrsMeta}{\aExpsMeta}}{\aExpsMeta} \mid
      \SKIP \mid      
  \\
    &
      \IF{\bExpsMeta}{\stmtsMeta}{\stmtsMeta} \mid
      \WHILE{\bExpsMeta}{\stmtsMeta} \mid
      \SEQ{\stmtsMeta}{\stmtsMeta} \mid
  % \\
  %   &
      \funCall 
      {\fidsMeta}
      {\aExpsMeta,\dots,\aExpsMeta}
      {\varsMeta,\dots,\varsMeta} 
\end{align*}
\normalsize
Here, $\arrsMeta$ is an array identifier,
and $\arrRef{\arrsMeta}{\aExpsMeta_1}$ is the expression
used to retrieve the element of the array $\arrsMeta$ at the index
$\aExpsMeta_1$.
In addition, $\varDecl{\varsMeta}$ is the declaration of the variable
$\varsMeta$,
$\arrDecl{\arrsMeta}{\numeralsMeta}$ is the declaration of the array
with identifier $\varsMeta$ and size $\numeralsMeta$,
$\ASSG{\arrRef{\arrsMeta}{\aExpsMeta_1}}{\aExpsMeta_2}$ is an
assignment of the result of $\aExpsMeta_2$ to the element of the array
$\arrsMeta$ indexed at $\aExpsMeta_1$, and
$\funCall {\fidsMeta} {\aExpsMeta_1,\dots,\aExpsMeta_m}
{\varsMeta_1,\dots,\varsMeta_n} $ is a call to the function with
identifier $\fidsMeta$ with arguments $\aExpsMeta_1$, \dots,
$\aExpsMeta_m$ and return variables $\varsMeta_1$, \dots,
$\varsMeta_n$.
If some argument $\aExpsMeta_i$ is an array, then it is passed by
reference in the call.
%Array arguments of function calls are passed by reference.

% In a function call
% $\funCall{\fidsMeta}{\LST{\aExpsMeta_1,\dots,\aExpsMeta_m}}
% {\varsMeta_1,\dots,\varsMeta_n}$, the array arguments among
% $\aExpsMeta_1,\dots,\aExpsMeta_m$ are passed by reference.

A \emph{program} in the extended While language is a mapping $\progsMeta$
from each function identifier $\fidsMeta$ to a triple
$(\LST{\pdsMeta_1,\dots,\pdsMeta_m},
\LST{\varsMeta_1,\dots,\varsMeta_n}, \stmtsMeta)$ or $\bot$.
Here, each $\pdsMeta_i$ ($i\in \{1,\dots,m\}$) is a parameter of the
function that is either a variable $\varsMeta$ or an array
$\arrsMeta$.
Each $\varsMeta_i$ ($i\in\{1,\dots,n\}$) is a return variable of the
function.
The $\stmtsMeta$ is the statement of the function.
If $\progsMeta(\fidsMeta)=\bot$, then there is no function defined for
the function identifier in the program.

\vspace*{-1.5ex}
\paragraph{Semantics. }
A \emph{state} $\ewhStatesMeta$ is a pair
$\ewhState{\storesMeta}{\locBoundsMeta}$.
Here,
$\storesMeta\in (\vars\cup\arrs\to
\optional{\ints})\cup(\ints\to\ints)$ is a \emph{store} that maps each
variable to an optional integer that is the value of the variable,
maps each array name to an optional integer representing the starting
location of the array, and maps each location to an integer that is
the value stored at the location.
In addition, $\locBoundsMeta\in \ints$ is the \emph{next fresh
  location} that can be used as the starting location of an array.
For a state $\ewhStatesMeta=\ewhState{\storesMeta}{\locBoundsMeta}$,
we write $\component{\ewhStatesMeta}{\storesMeta}$ to refer to
the store $\storesMeta$,
write $\component{\ewhStatesMeta}{\locBoundsMeta}$ to refer to
the next fresh location $\locBoundsMeta$ for arrays,
write $\ewhStatesMeta(a)$ for $\storesMeta(a)$, and 
write $\ewhStatesMeta[a\mapsto b]$ for $\storesMeta[a\mapsto b]$.

\begin{figure}[t]
  \small
  \begin{align*}
    &
      \RULE
      {(\ewhConfig{\SKIP}{\ewhStatesMeta}{\progsMeta},
      \ewhStatesMeta)}  
      {\,} \\%\\[-3ex]
    &
      \RULE
      {(\ewhConfig{\varDecl{\varsMeta}}{\ewhStatesMeta} 
      {\progsMeta},
      \storeUpd{\ewhStatesMeta}{\varsMeta}{0}
      )} 
      {\, } \quad
      \mrm{if}~\ewhStatesMeta(\varsMeta)=\bot
    \\%\\[-3ex]
    &
      \RULE
      {(\ewhConfig{\arrDecl{\arrsMeta}{\numeralsMeta}}
      {\ewhState{\storesMeta}{\locBoundsMeta}}{\progsMeta}, 
      \ewhRConfig
      {\storesMeta
      [\arrsMeta\mapsto \locBoundsMeta]} 
      {\locBoundsMeta+ 
      \ewhNEval{\numeralsMeta}})  
      }
      {\, }
      \quad \mrm{if}~\ewhStatesMeta(\arrsMeta)=\bot \\%\\[-3ex]
    &
      \RULE
      {
      (
      \ewhConfig{\ASSG{\varsMeta}{\aExpsMeta}}
      {\ewhStatesMeta}{\progsMeta},
      \storeUpd{\ewhStatesMeta}{\varsMeta}{\intsMeta}
      ) 
      }
      {\,}
      \quad
      \mrm{if}~
      \ewhAEval{\aExpsMeta}\ewhStatesMeta=\intsMeta \land
      \ewhStatesMeta(\varsMeta)\neq\bot
    \\%\\[-3ex]
    &
      \RULE
      {
      (\ewhConfig{\ASSG{\arrRef{\arrsMeta}{\aExpsMeta_1}}{\aExpsMeta_2}}
      {\ewhStatesMeta}{\progsMeta},
      \ewhStatesMeta') 
      }
      {\,}
    \\
    & \qquad
      \mrm{if}~
      \exists \locsMeta: \exists \intsMeta_1,\intsMeta_2:
      (\ewhStatesMeta(\arrsMeta)=\locsMeta %\ge 0
      \,\land\,
      \ewhAEval{\aExpsMeta_1}\ewhStatesMeta
      =\intsMeta_1\,\land\, \intsMeta_1 \ge 0 \,\land\,
      \ewhAEval{\aExpsMeta_2}\ewhStatesMeta 
      =\intsMeta_2 \,\land\, \\[-.5ex]
    & \qquad\qquad\qquad\qquad ~~\,
      \locsMeta+\intsMeta_1<\component{\ewhStatesMeta}{\locBoundsMeta} \land
      \ewhStatesMeta'=
      \heapUpd{\ewhStatesMeta}{\locsMeta+\intsMeta_1}
      {\intsMeta_2} )
      \\\\[-3ex]
    &
      \RULE
      {(\ewhConfig{\IF{\bExpsMeta}{\stmtsMeta_1}{\stmtsMeta_2}}
      {\ewhStatesMeta}{\progsMeta},
      \ewhStatesMeta')} 
      {(\ewhConfig{\stmtsMeta_1}{\ewhStatesMeta}
      {\progsMeta}, \ewhStatesMeta')} 
      \quad \mrm{if}~
      \ewhBEval{\bExpsMeta}\ewhStatesMeta=\TT  \\%\\[-3ex]
    &
      \RULE
      {(\ewhConfig{\IF{\bExpsMeta}{\stmtsMeta_1}{\stmtsMeta_2}}
      {\ewhStatesMeta}{\progsMeta}, \ewhStatesMeta')}
      {(\ewhConfig{\stmtsMeta_2}{\ewhStatesMeta}{\progsMeta}, 
      \ewhStatesMeta')} 
      \quad \mrm{if}~
      \ewhBEval{\bExpsMeta}\ewhStatesMeta=\FF \\%\\[-3ex]
    &
      \RULE
      {
      (\ewhConfig{\WHILE{\bExpsMeta}{\stmtsMeta}}
      {\ewhStatesMeta}{\progsMeta},
      \ewhStatesMeta')
      }
      {
      (\ewhConfig{\stmtsMeta}{\ewhStatesMeta}{\progsMeta},
      \ewhStatesMeta''),
      (\ewhConfig{\WHILE{\bExpsMeta}{\stmtsMeta}}
      {\ewhStatesMeta''}{\progsMeta},
      \ewhStatesMeta') 
      } 
    ~~ \mrm{if}~
      \ewhBEval{\bExpsMeta}\ewhStatesMeta=\TT
    \\%\\[-3ex]
    &
      \RULE
      {
      (\ewhConfig{\WHILE{\bExpsMeta}{\stmtsMeta}}
      {\ewhStatesMeta}{\progsMeta},
      \ewhStatesMeta) 
      }
      {\,}
      \quad \mrm{if}~
      \ewhBEval{\bExpsMeta}\ewhStatesMeta=\FF 
    \\%\\[-3ex]
    &
      \RULE
      {
      (\ewhConfig{\SEQ{\stmtsMeta_1}{\stmtsMeta_2}}
      {\ewhStatesMeta}{\progsMeta}, \ewhStatesMeta') 
      }
      {
      (\ewhConfig{\stmtsMeta_1}{\ewhStatesMeta}{\progsMeta},
      \ewhStatesMeta''),
      (\ewhConfig{\stmtsMeta_2}{\ewhStatesMeta''}{\progsMeta},
      \ewhStatesMeta') 
      }
    \\%\\[-3ex]
    &
      \RULE
      {
      (\ewhConfig
      {\funCall{\fidsMeta}{\aExpsMeta_1,\dots,\aExpsMeta_m}
      {\varsMeta'_1,\dots,\varsMeta'_n}}
      {\ewhStatesMeta}{\progsMeta},
      \ewhState{\storesMeta'}{\component{\ewhStatesMeta}{\locBoundsMeta}}) 
      }
      {
      (\ewhConfig{\stmtsMeta}{\ewhState{\storesMeta''}
      {\component{\ewhStatesMeta}{\locBoundsMeta}}}{\progsMeta},
      \ewhStatesMeta'_0)
      } \\
    &
      \qquad
      \mrm{if}~
      \exists \valsMeta_1,\dots,\valsMeta_m: \\
    &\qquad\qquad 
      \progsMeta(\fidsMeta)=
      (\LST{\pdsMeta_1,\dots,\pdsMeta_m},  
      \LST{\varsMeta_1,\dots,\varsMeta_n},
      \stmtsMeta)
      \,\land\,
      \ewhAEval{\aExpsMeta_1}\ewhStatesMeta=\valsMeta_1 \land
      \dots \land 
      \ewhAEval{\aExpsMeta_m}\ewhStatesMeta=\valsMeta_m \,\land\,
      % \sizeOf{\listsOfNamesMeta}=m \,\land\, 
    \\
    &\qquad\qquad
      \callIniSt{\component{\ewhStatesMeta}{\storesMeta}} 
      {\LST{\pdsMeta_1,\dots,\pdsMeta_m}} 
      {\LST{\valsMeta_1,\dots,\valsMeta_m}}
      {\LST{\varsMeta_1,\dots,\varsMeta_n}}= \storesMeta''
      \,\land\, 
      \\
    &\qquad\qquad
      \callFinSt
      {\component{\ewhStatesMeta}{\storesMeta}}
      {\component{\ewhStatesMeta'_0}{\storesMeta}}
      {\LST{\varsMeta_1,\dots,\varsMeta_n}}
      {\LST{\varsMeta'_1,\dots,\varsMeta'_n}}
      =\storesMeta'
      % \ewhStatesMeta
      % [\varsMeta'_1\mapsto
      % \ewhStatesMeta'_0(\varsMeta_1),
      % \dots,
      % \varsMeta'_n\mapsto
      % \ewhStatesMeta'_0(\varsMeta_n)
      % ]
      % =\storesMeta'
  \end{align*}
  \normalsize
  \vspace*{-5ex}
  \caption{The semantic rules for the statements of the extended While language}
  \vspace*{-2ex}
  \label{fig:sem-rules-ext-while-language}
\end{figure}

We define the big-step semantics of the extended While language by
defining the predicate $\mi{rule}$ as in
\cref{fig:sem-rules-ext-while-language}.
According to the rule for the array declaration
$\arrDecl{\arrsMeta}{\numeralsMeta}$, the array identifier should be
mapped to the undefined location $\bot$ before the declaration.
% the array identifier should not
% be occupied (mapped to $\bot$ by the state) before the declaration. 
%
This array identifier is then associated to the next fresh location
$\locBoundsMeta$ that can be used for arrays, and the next fresh
location for arrays is incremented by the size of the array after the
declaration.
According to the rule for the assignment
$\ASSG{\arrRef{\arrsMeta}{\aExpsMeta_1}}{\aExpsMeta_2}$, the location
of the target array element should not surpass the boundary as given
by the next fresh location $\locBoundsMeta$.
The result of the right-hand side is then placed at this location.
In the rule for function calls, the auxiliary function
$\mi{call}\mbox{-}\mi{ini}$ is used to initialize the store for the
execution of the callee.
The resulting initial store $\storesMeta''$ for the callee maps the
parameters of the callee to the values of the corresponding arguments,
and maps each return variable of the callee to $0$.
% This function sets the parameters to the values of their corresponding
% arguments, and sets the return variables to $0$.
%
In the same rule, the auxiliary function $\mi{call}\mbox{-}\mi{fin}$
is used to finalize the store after the execution of the callee.
The resulting store $\storesMeta'$ maps the caller's variables that
receive the return values to the values of the callee's return
variables.
In addition, $\storesMeta'$ maps the memory locations according to
the store $\component{\ewhStatesMeta'_0}{\storesMeta}$ reached after
the execution of the callee's statement.
Hence, the effects of the callee on the arrays passed in by reference
are recorded.
On the other hand, the next fresh location
$\component{\ewhStatesMeta}{\locBoundsMeta}$ for arrays before the
call is kept after the call returns.
Hence, any arrays that are stack-allocated (at their declarations) in
the callee are discarded.
% This function sets the variables receiving the return values according
% to the values of the corresponding return variables of the callee. 

In \cref{fig:sem-rules-ext-while-language}, $\mathcal{A}$ and
$\mathcal{B}$ are evaluation functions for arithmetic expressions and
Boolean expressions, respectively.
The detailed definitions for these two functions, and for the
functions $\mi{call}\mbox{-}\mi{ini}$ and $\mi{call}\mbox{-}\mi{fin}$,
are given in \cref{asec:supp-ext-while}.

\vspace*{-1.5ex}
\subsubsection*{Array-Merging Program and its Verification}

\begin{figure}[t]
  \small
  \begin{align*}
    \mergeProg ~\defEq~
    &
    [\,
    \mergeFId\mapsto (\LST{\SArr,\TArr,\pgVar{i},\pgVar{m},\pgVar{n}},
    \emptyList, \mergeStmt)% , 
      % \,
      % \msortFId\mapsto (\LST{\SArr,\TArr,\lVar,\hVar}, \LST{}, \msortStmt)
    \,]
    \\\\[-2ex]
  \mergeStmt
  ~\defEq ~
    &
      \varDecl{\jVar}; \varDecl{\kVar}; 
    \ASSG{\pgVar{j}}{\pgVar{m}+1}; \ASSG{\pgVar{k}}{\pgVar{i}};
    \mgLoopStmt;
    \tailLoopStmt{\pgVar{i}}{\pgVar{m}};
    \tailLoopStmt{\pgVar{j}}{\pgVar{n}} \\
  \mgLoopStmt
  ~\defEq ~
  &
    \mathsf{while}
    ~\andExpr{\pgVar{i}\le\pgVar{m}}{\pgVar{j}\le\pgVar{n}}~
    \mathsf{do}~( \\ 
  &
    \quad
    (\mathsf{if}~
    \arrRef{\SArr}{\pgVar{i}}\le \arrRef{\SArr}{\pgVar{j}}
    ~\mathsf{then}~
    \ASSG{\arrRef{\TArr}{\pgVar{k}}}{\arrRef{\SArr}{\pgVar{i}}};
    \ASSG{\mathsf{i}}{\mathsf{i}+1} 
    ~\mathsf{else}~ 
    \ASSG{\arrRef{\TArr}{\pgVar{k}}}{\arrRef{\SArr}{\pgVar{j}}};
    \ASSG{\mathsf{j}}{\mathsf{j}+1}); \\
  &
    \quad \ASSG{\pgVar{k}}{\pgVar{k}+1}\, ) \\
%  & ) \\
  \tailLoopStmt{\pgVar{i}}{\pgVar{m}}
  ~\defEq~
  &
    \mathsf{while}~\pgVar{i}\le\pgVar{m}~\mathsf{do}~
    (\ASSG{\arrRef{\TArr}{\pgVar{k}}}{\arrRef{\SArr}{\pgVar{i}}};
    \ASSG{\mathsf{i}}{\mathsf{i}+1};
    \ASSG{\mathsf{k}}{\mathsf{k}+1}) \\
  \tailLoopStmt{\pgVar{j}}{\pgVar{n}}
  ~\defEq~
  &
    \mathsf{while}~\pgVar{j}\le\pgVar{n}~\mathsf{do}~
    (\ASSG{\arrRef{\TArr}{\pgVar{k}}}{\arrRef{\SArr}{\pgVar{j}}};
    \ASSG{\mathsf{j}}{\mathsf{j}+1}; 
    \ASSG{\mathsf{k}}{\mathsf{k}+1})
    \\[-5ex]
  %   \\\\[-2ex]
  % \msortStmt ~\defEq~ 
  % &
  %   \varDecl{\mVar}; 
  %   \arrDecl{\tmpArr}{100}; \\
  % &
  %   \mathsf{if}~\lVar=\hVar~\mathsf{then}~
  %   \ASSG{\arrRef{\TArr}{\lVar}}{\arrRef{\SArr}{\lVar}} \\
  % &
  %   \mathsf{else}~\\
  % &
  %   \quad \ASSG{\pgVar{m}}{(\lVar+\hVar)/2}; 
  %    \funCall{\msortFId}{\SArr,\tmpArr,\lVar,\mVar}{}; 
  %    \funCall{\msortFId}{\SArr,\tmpArr,\mVar+1,\hVar}{}; \\
  % &
  %   \quad \funCall{\mergeFId}{\tmpArr,\TArr,\lVar,\mVar,\hVar}{}    
  \end{align*}
  \normalsize
  \vspace*{-2ex}
  \caption{The program $\mergeProg$ that merges sorted array fragments}
  \vspace*{-3ex}
  \label{fig:prog-merge-array}
\end{figure}

The program $\mergeProg$ as shown in \cref{fig:prog-merge-array}
merges the elements in two sorted fragments of an array $\SArr$ into
one sorted fragment in a different array $\TArr$.

The only function in this program is $\mergeFId$. 
Formally, this function is the triple
$(\LST{\SArr,\TArr,\pgVar{i},\pgVar{m},\pgVar{n}},
\emptyList, \mergeStmt)$.
The parameters $\pgVar{i}$ and $\pgVar{m}$ represent the initial and final
index, respectively, for the first fragment of the array $\SArr$ 
participating in the merger. 
The second fragment participating in the merger is from the
index represented by $\pgVar{m}+1$ to the index represented by
$\pgVar{n}$ in the same array $\SArr$.
The target array fragment of the merger is from the index represented
by $\iVar$ to the index represented by $\nVar$, in the array $\TArr$.
% The target array of the merger is $\TArr$. 

% The function merge the sorted list represented by the elements indexed
% from $\pgVar{i}$ to $\pgVar{m}$ of the array $\SArr$ and the sorted
% list represented by the elements indexed from $\pgVar{m}+1$ to
% $\pgVar{n}$ into a sorted list represented by the elements indexed
% from $\pgVar{i}$ to $\pgVar{n}$ of array $\TArr$.

% As shown in \cref{fig:prog-merge-sort}, the merge sort program
% consists of two functions, $\mergeFId$ and $\msortFId$.

% The function $\mergeFId$ is the triple
% $(\LST{\SArr,\TArr,\pgVar{i},\pgVar{m},\pgVar{n}},
% \emptyList, \mergeStmt)$.
% %
% The intended functionality is to merge the sorted list represented by
% the elements indexed from $\pgVar{i}$ to $\pgVar{m}$ of the array
% $\SArr$ and the sorted list represented by the elements indexed from
% $\pgVar{m}+1$ to $\pgVar{n}$ into a sorted list represented by the
% elements indexed from $\pgVar{i}$ to $\pgVar{n}$ of array $\TArr$. 

% The function $\msortFId$ is the triple
% $(\LST{\SArr,\TArr,\lVar,\hVar}, \LST{}, \msortStmt)$.
% %
% The intended functionality is to sort the fragment indexed from
% $\lVar$ to $\hVar$ in the array $\SArr$ and put the sorting result in
% $\TArr$, starting from the index $\lVar$.

For the specification of the program, we use a few pieces of auxiliary
notation.
We write $\arrFrag{\arrsMeta}{l}{h}$ for a triple $(\arrsMeta,l,h)$
that represents the fragment of the array $\arrsMeta$ from the index
$l$ to the index $h$. 
We write $\listOfArrFrag{\arrFrag{\arrsMeta}{l}{h}}{\ewhStatesMeta}$
for the list
$\LST{\ewhStatesMeta(\locsMeta+l), \dots,
  \ewhStatesMeta(\locsMeta+h)}$ where
$\locsMeta=\ewhStatesMeta(\arrsMeta)$, i.e., the list of elements of
the array $\arrsMeta$ from the index $l$ to the index $h$.
We write $\occ{\LST{\intsMeta_1,\dots,\intsMeta_n}}$ for the function
$\occFuncsMeta$ mapping each integer $\intsMeta$ to the number of
occurrences of $\intsMeta$ in the list
$\LST{\intsMeta_1,\dots,\intsMeta_n}$ of integers.
For two such functions $\occFuncsMeta_1$ and $\occFuncsMeta_2$, we
write $\occAdd{\occFuncsMeta_1}{\occFuncsMeta_2}$ for the function
$\lambda
\intsMeta. \occFuncsMeta_1(\intsMeta)+\occFuncsMeta_2(\intsMeta)$.
We write $\sorted{\LST{\intsMeta_1,\dots,\intsMeta_n}}$ to express
that the list $\LST{\intsMeta_1,\dots,\intsMeta_n}$ of integers is
sorted in ascending order.
We write $\sepArrFrag{\arrFrag{X}{l_1}{h_1}}{\arrFrag{Y}{l_2}{h_2}}{\ewhStatesMeta}$ 
% $\separated{\arrsMeta_1}{l_1}{h_1}{\arrsMeta_2}{l_2}{h_2}{\ewhStatesMeta}$
to express that the elements of the array $X$ from the index
$l_1$ to the index $h_1$ occupy a separate memory area from that
occupied by the elements of the array $Y$ from the index
$l_2$ to the index $h_2$, in the state $\ewhStatesMeta$. 
In addition, we write
$\preserved{\ewhStatesMeta}{\ewhStatesMeta'}
{\varOrArrFragMeta_1,\dots,\varOrArrFragMeta_n}$ to express for each
$i\in\{1,\dots,n\}$, the value of each $\varOrArrFragMeta_i$ is the
same in the states $\ewhStatesMeta$ and $\ewhStatesMeta'$.
Here, $\varOrArrFragMeta_i$ can be a variable $\varsMeta$ or an array
fragment $\arrFrag{\arrsMeta}{l}{h}$.
In the latter case, that the value of $\arrFrag{\arrsMeta}{l}{h}$ is
the same in the two states means
$\forall i: l \le i \le h\Rightarrow
\ewhStatesMeta(\ewhStatesMeta(\arrsMeta)+i)=
\ewhStatesMeta'(\ewhStatesMeta'(\arrsMeta)+i)$.

For the program $\mergeProg$, we devise the specification $\specMSort$.
We denote the starting index for the first source array fragment
in $\SArr$ as well as for the target array fragment in $\TArr$ by $\lVal$.
We use $\lVal$ as a global parameter in the specification.
%
% For the program $\mergeProg$, we devise the specification $\specMSort$
% with the constant $\lParam$ as a global parameter.

We specify the function $\mergeFId$ as

\vspace*{-3ex}
\small
\begin{align*}
  & 
  \specMSort(
    \ewhConfig{\funCall{\mergeFId}
    {\arrsMeta,\yarrsMeta,\aExpl,\aExpm,\aExph}{}} 
    {\ewhStatesMeta}{\mergeSortProg}
    ) \,\defEq \, \\
  & \quad
    \begin{aligned}
      &
      \{
      \ewhStatesMeta' \mid
      \occ{\listOfArrFrag{\arrFrag{\arrsMeta}{\lVal}{\hVal}}{\ewhStatesMeta}}
      =
      \occ{\listOfArrFrag{\arrFrag{\yarrsMeta}{\lVal}{\hVal}}{\ewhStatesMeta'}}
      % (\forall \intsMeta:
      % \occ{\listOfArrFrag{\arrsMeta_1}{\lVal}{\hVal}{\ewhStatesMeta}}
      % {\intsMeta} = 
      % \occ{\listOfArrFrag{\arrsMeta_2}{\lVal}{\hVal}{\ewhStatesMeta'}}
      % {\intsMeta} )
      \,\land\,
      % \\
      % &\qquad~
      \sorted{\listOfArrFrag{\arrFrag{\yarrsMeta}{\lVal}{\hVal}}{\ewhStatesMeta'}}
      \, \}
      \\\\[-2.5ex]
      &
      \mrm{if}~
      \ewhAEval{\aExpl}\ewhStatesMeta \!=\! l\,\land\,
      0\!\le\! \lVal\! \le\! \mVal \!<\! \hVal \,\land\, 
      \sorted{
        \listOfArrFrag{\arrFrag{\arrsMeta}{\lVal}{\mVal}}
        {\ewhStatesMeta}}
      \land
      \sorted{
        \listOfArrFrag{\arrFrag{\arrsMeta}{\mVal+1}{\hVal}}
        {\ewhStatesMeta}} 
      \land 
      \sepArrFrag{\arrFrag{\arrsMeta}{\lVal}{\hVal}}{\arrFrag{\yarrsMeta}{\lVal}{\hVal}}
      {\ewhStatesMeta}
      \\
      &
      \mrm{where}~
      % \lVal=\ewhAEval{\aExpl}\ewhStatesMeta \,\land\,
      \mVal=\ewhAEval{\aExpm}\ewhStatesMeta \,\land\,
      \hVal=\ewhAEval{\aExph}\ewhStatesMeta
    \end{aligned}
\end{align*}
\normalsize
This specification says that if we call the function $\mergeFId$
with two array identifiers $\arrsMeta$ and $\yarrsMeta$, and
expressions $\aExpl$, $\aExpm$, $\aExph$ that evaluate to
$\lVal$, $\mVal$ and $\hVal$, such that

\vspace*{-3ex}
\begin{itemize}
\item $0\le\lVal\le\mVal<\hVal$ holds,  
\item the array fragments $\arrFrag{\arrsMeta}{\lVal}{\mVal}$
  and $\arrFrag{\arrsMeta}{\mVal+1}{\hVal}$ are sorted in the pre-state, 
\item the array fragments $\arrFrag{\arrsMeta}{\lVal}{\mVal}$
  and $\arrFrag{\arrsMeta}{\mVal+1}{\hVal}$ are separated in the pre-state,
  \vspace*{-3ex}
\end{itemize}
then the number of occurrences of each integer in the target array
fragment $\arrFrag{\yarrsMeta}{\lVal}{\hVal}$ in the post-state is the
same as its number of occurrences in the source array fragment
$\arrFrag{\arrsMeta}{\lVal}{\hVal}$ in the pre-state, and the target
array fragment $\arrFrag{\yarrsMeta}{\lVal}{\hVal}$ is sorted in
ascending order in the post-state.

The core part of the function $\mergeFId$ is the loop statement
$\mgLoopStmt$ (see \cref{fig:prog-merge-array}).
We specify this loop as

\vspace*{-2ex}
\small
\begin{align*}
  &
  \specMSort(\ewhConfig{\mgLoopStmt}{\ewhStatesMeta}{\mergeSortProg}) 
    \,\defEq\, \\
  &\quad 
  \begin{aligned}
  &
    \{
    \ewhStatesMeta' \mid
    (i \le \ewhStatesMeta'(\iVar)=\mVal+1 \land
    \jVal \le \ewhStatesMeta'(\jVar) \le \nVal \lor 
    j \le \ewhStatesMeta'(\jVar)=\nVal+1 \land
    \iVal \le \ewhStatesMeta'(\iVar) \le \mVal) \,\land\, \\%[-.5ex]
  &
    \qquad~
    \ewhStatesMeta'(\kVar)=
    \kVal + \ewhStatesMeta'(\iVar) - \iVal
    + \ewhStatesMeta'(\jVar) - \jVal \,\land \,
    \preserved{\ewhStatesMeta}{\ewhStatesMeta'}
    {\mVar, \nVar, \SArr, \TArr, \arrFrag{\SArr}{\lVal}{\nVal},
    \arrFrag{\TArr}{\lVal}{\kVal-1}} \,\land\, 
    % \ewhStatesMeta'(\mVar) = \mVal \,\land\,
    % \ewhStatesMeta'(\nVar) = \nVal \,\land\,
    % \ewhStatesMeta'(\SArr)=\ewhStatesMeta(\SArr)\,\land\,
    % \ewhStatesMeta'(\TArr)=\ewhStatesMeta(\TArr) \,\land\,
    % \\%[-.5ex]
    % &
    % \qquad~
    % \listOfArrFrag{\SArr}{\lParam}{n}{\ewhStatesMeta}=
    % \listOfArrFrag{\SArr}{\lParam}{n}{\ewhStatesMeta'} \,\land\,
    % \listOfArrFrag{\TArr}{\lParam}{k-1}{\ewhStatesMeta}=
    % \listOfArrFrag{\TArr}{\lParam}{k-1}{\ewhStatesMeta'}\,\land\,
    \\
  &
  \qquad~
  \occAdd{
    \occ{\listOfArrFrag{\arrFrag{\SArr}{\iVal}{\ewhStatesMeta'(\iVar)-1}}
      {\ewhStatesMeta}}
  }
  {
    \occ{\listOfArrFrag{\arrFrag{\SArr}{\jVal}{\ewhStatesMeta'(\jVar)-1}}
      {\ewhStatesMeta}} 
  }=
  \occ{\listOfArrFrag{\arrFrag{\TArr}{\kVal}{\ewhStatesMeta'(\kVar)-1}}
    {\ewhStatesMeta'}}
  % \left (
  %   \begin{aligned}
  %     \forall \intsMeta:
  %     &\,
  %     \occ{\listOfArrFrag{\SArr}{\iVal}{\ewhStatesMeta'(\iVar)-1}
  %       {\ewhStatesMeta}}
  %     {\intsMeta} +
  %     \occ{\listOfArrFrag{\SArr}{\jVal}{\ewhStatesMeta'(\jVar)-1}
  %       {\ewhStatesMeta}}
  %     {\intsMeta} 
  %     = \\
  %     &\,
  %     \occ{\listOfArrFrag{\TArr}{\kVal}{\ewhStatesMeta'(\kVar)-1}
  %       {\ewhStatesMeta'}}{\intsMeta}
  %   \end{aligned}
  %   \right )
  %  \,\land\, \\%[-.5ex]
  %  &\qquad~
  \,\land\, 
  \sorted{
    \listOfArrFrag
    {\arrFrag{\TArr}{\lParam}{\ewhStatesMeta'(\kVar)-1}}
    {\ewhStatesMeta'}
  }
    % \sorted{\listOfArrFrag{\SArr}{\lParam}{\ewhStatesMeta'(\mVar)}{\ewhStatesMeta'}}
    % \,\land\,
    % \sorted{\listOfArrFrag{\SArr}{\lParam}{\ewhStatesMeta'(\mVar)+1}{\ewhStatesMeta'}}    
  \,\land\, \\%[-.5ex]
  &\qquad~
   (
  \ewhStatesMeta'(\iVar)\le m \land
  \ewhStatesMeta'(\kVar)\ge l+1 \Rightarrow 
  \ewhAEval{\arrRef{\SArr}{\iVar}}\ewhStatesMeta' \!\ge\!
  \ewhAEval{\arrRef{\TArr}{\kVar\!-\!1}}\ewhStatesMeta') \,\land\,  
  % (%i<
  % \ewhStatesMeta'(\iVar)\le m \land
  % \ewhStatesMeta'(\jVar)\ge m+2 \Rightarrow 
  % \ewhAEval{\arrRef{\SArr}{\iVar}}\ewhStatesMeta' \!\ge\!
  % \ewhAEval{\arrRef{\SArr}{\jVar\!-\!1}}\ewhStatesMeta') \,\land\,
  \\%[-.5ex]
  &\qquad~
  (
  \ewhStatesMeta'(\jVar)\le n \land
  \ewhStatesMeta'(\kVar) \ge l+1 \Rightarrow 
   \ewhAEval{\arrRef{\SArr}{\jVar}}\ewhStatesMeta' \!\ge\!
   \ewhAEval{\arrRef{\TArr}{\kVar\!-\!1}}\ewhStatesMeta')  
  % (%m+1<
  % \ewhStatesMeta'(\jVar)\le n \land
  % \ewhStatesMeta'(\iVar) \ge i+1 \Rightarrow 
  %  \ewhAEval{\arrRef{\SArr}{\jVar}}\ewhStatesMeta' \!\ge\!
  %  \ewhAEval{\arrRef{\SArr}{\iVar\!-\!1}}\ewhStatesMeta')
    \,
    \}
    \\\\[-2.5ex]
    &
    \mrm{if}~0\le \lParam \le \iVal \le \mVal < \jVal \le \nVal \,\land\,
    k = i+j-m-1 \,\land\,
    \\
    &\quad
    (k \ge l+1 \,\Rightarrow\,
    \ewhAEval{\arrRef{\SArr}{\iVar}}\ewhStatesMeta \!\ge\!
    \ewhAEval{\arrRef{\TArr}{\kVar\!-\!1}}\ewhStatesMeta
    \land
    \ewhAEval{\arrRef{\SArr}{\jVar}}\ewhStatesMeta \!\ge\!
    \ewhAEval{\arrRef{\TArr}{\kVar\!-\!1}}\ewhStatesMeta)\,\land\, 
    \\
    &\quad 
    \sorted{\listOfArrFrag{\arrFrag{\SArr}{\iVal}{\mVal}}{\ewhStatesMeta}}
    \,\land\,
    \sorted{\listOfArrFrag{\arrFrag{\SArr}{\jVal}{\nVal}}{\ewhStatesMeta}}
    \,\land\,
    \sorted{\listOfArrFrag{\arrFrag{\TArr}{\lParam}{k-1}}{\ewhStatesMeta}}
    \,\land\, 
   % \\%[-.5ex]
   % &\quad
      \sepArrFrag
      {\arrFrag{\SArr}{\lVal}{\nVal}}
      {\arrFrag{\TArr}{\lVal}{\nVal}}
      {\ewhStatesMeta} 
    % \separated{\SArr}{\iVal}{\mVal}{\TArr}{\kVal}{\nVal}{\ewhStatesMeta}
    % \,\land\,
    % \separated{\SArr}{\jVal}{\nVal}{\TArr}{\kVal}{\nVal}{\ewhStatesMeta}
  \\
  & \mrm{where}~
    \iVal = \ewhStatesMeta(\iVar)\,\land \,
    \jVal = \ewhStatesMeta(\jVar)\,\land \,
    \kVal = \ewhStatesMeta(\kVar)\,\land \,
    \mVal = \ewhStatesMeta(\mVar)\,\land \,
    \nVal = \ewhStatesMeta(\nVar)
  \end{aligned}
\end{align*}
\normalsize
In the specification, we are concerned with pre-states in which either
the overall loop is yet to be executed, or some rounds of the loop
have been completed and some further rounds are to be executed.
We constrain these pre-states with a few further conditions.
One of these conditions states that the elements with indexes $\iVar$
and $\jVar$ that are to be compared in the next round are both greater
than or equal to the last element that has been set in the target
array fragment.
For each pre-state that satisfies all the conditions in the ``if''
part, several conditions are asserted for the potential post-state
$\ewhStatesMeta'$.
A key condition here says that the two fragments
{\small$\arrFrag{\SArr}{\iVal}{\ewhStatesMeta'(\iVar)-1}$} and
{\small$\arrFrag{\SArr}{\jVal}{\ewhStatesMeta'(\jVar)-1}$} in the
source array that are scanned between the reaching of the pre-state
and the post-state agree with the fragment
{\small$\arrFrag{\TArr}{\kVal}{\ewhStatesMeta'(\kVar)-1}$} that is
filled between the reaching of the pre-state and the post-state.
Another key condition says that the fragment
{\small$\arrFrag{\TArr}{\lVal}{\ewhStatesMeta'(\kVar)-1}$} of the
target array that is already filled in the post-state for the loop is
sorted in ascending order.

% For the pre-states that satisfy all the conditions formulated in the
% ``if'' part, several conditions are asserted for the post-state
% $\ewhStatesMeta'$ after the loop exits.
% %
% A key condition here says that the two fragments in the source array
% that are yet to be scanned agree with the fragment in the target array
% that is yet to be filled, in the number of occurrences of each value.
% %
% Another key condition says that the fragment of the target array that
% is already filled is sorted in ascending order.

Without specification inference, the two remaining loops in the
array-merging program also need to be explicitly specified.
The specification of these two loops is much 
%considerably
less involved than that for the first loop, and it is defered to
\cref{asec:supp-ext-while}. 
%
%We defer these two specifications to the appendix.
%
With the technique of \cref{sec:method}, the validity of 
%the specification
$\specMSort$ can be established. 
\begin{theorem}\label{thm:while-ext-merge-spec-valid} 
  It holds that $\specValid{\specMSort}$. 
\end{theorem}
With the help of \cref{thm:soundness}, the proof requires no induction
for reasoning about the loops.
This proof boils down to symbolic execution with the help of a series
of auxiliary lemmas about the memory layout. 
% The proof of this theorem is performed using \cref{thm:soundness}. 
% %
% The use of \cref{thm:soundness} factors out the inductions needed for
% reasoning about the loops.
% %
% Thus, the proof consists mainly of symbolic execution and the
% reasoning about the memory layout using a series of auxiliary lemmas.

\begin{remark}
  The global parameter $\lParam$ in the specification $\specMSort$
  relates the auxiliary information about calls to $\mergeFId$ and
  about the loops in this function.
  % the body of 
  %
  % To ease understanding,
  The role of $\lParam$ can be compared to that of a logical variable
  in a concrete program logic.
  Such global parameters are captured in the Coq formalization by an
  explicit argument in the specifications.
  % The way such global parameters are captured in the Coq formalization
  % is by an explicit argument in the specifications.
  %
  The type of this argument can be instantiated according to the needs
  in verifying each specific program.
  The verification of a program is required to go through for all
  possible values of this argument.
\end{remark}

% \subsubsection*{Proof Outline for Array-Merging Program}

%%% Local Variables:
%%% mode: latex
%%% TeX-master: "main"
%%% End:

\vspace*{-2ex}
\subsection{Eager Functional Language and List-Merging Program}\label{ssec:functional}
%\vspace*{-1ex}

\subsubsection*{Eager Functional Language}

The language considered in this section is a fragment of the eager
functional language as discussed in \cite{Reynolds1998}.
The expressions and canonical forms of this language are given in
\cref{fig:exprs-and-canonical-forms}. 
%
%They are explained in further details below. 

\begin{figure}[t]
     \centering
     \begin{subfigure}[b]{0.48\textwidth}
       \small
       \begin{align*}
         \fExprsMeta \,\synDefEq~
         &
           \fNumConstsMeta \mid \trueConst \mid
           \falseConst \mid  \\
         &
           \FADD{\fExprsMeta}{\fExprsMeta} \mid
           \FSUB{\fExprsMeta}{\fExprsMeta} \mid
           \FMUL{\fExprsMeta}{\fExprsMeta} \mid
           \FDIV{\fExprsMeta}{\fExprsMeta} \mid \\
         &
           \FEQ{\fExprsMeta}{\fExprsMeta} \mid
           \FLT{\fExprsMeta}{\fExprsMeta} \mid
           \FNEG{\fExprsMeta} \mid
           \FAND{\fExprsMeta}{\fExprsMeta} \mid \\
         &
           \IFTHENELSE{\fExprsMeta}{\fExprsMeta}{\fExprsMeta} \mid \\
         &
           \NIL \mid \CONCAT{\fExprsMeta}{\fExprsMeta} \mid
           \LCASE{\fExprsMeta}{\fExprsMeta}{\fExprsMeta} \mid \\ 
         &
           \fVarsMeta \mid
           \FAPP{\fExprsMeta}{\fExprsMeta} \mid
           \LAM{\fVarsMeta}{\fExprsMeta} \mid 
           % \LETIN{\fVarsMeta}{\fExprsMeta}{\fExprsMeta} \mid
           \LETRECIN{\fVarsMeta}{\LAM{\fVarsMeta'}{\fExprsMeta}}{\fExprsMeta}
       \end{align*}
       \normalsize
     \end{subfigure}
     \hfill
     \begin{subfigure}[b]{0.48\textwidth}
       \small
       \begin{align*}
         \cfmsMeta \,\synDefEq\,
         &
           \intCfmsMeta \mid \boolCfmsMeta \mid
           \funCfmsMeta \mid \lstCfmsMeta \\
         \intCfmsMeta \,\synDefEq\,
         &
           \ldots \mid -2 \mid -1 \mid 0 \mid 1 \mid 2 \mid \ldots \\
         \boolCfmsMeta \,\synDefEq\,
         &
           \trueConst \mid \falseConst \\
         \funCfmsMeta \,\synDefEq\,
         &
           \LAM{\fVarsMeta}{\fExprsMeta} \\
         \lstCfmsMeta \,\synDefEq\,
         &
           \NIL \mid \CONCAT{\cfmsMeta}{\cfmsMeta} 
       \end{align*}
       \normalsize
     \end{subfigure}
     \caption{The expressions and canonical forms of the
       eager functional language}
     \vspace*{-3ex}
   \label{fig:exprs-and-canonical-forms}
\end{figure}

\vspace*{-1ex}
\paragraph{Syntax. } 
%
% \vspace*{-2ex}
% \small
% \begin{align*}
%   \fExprsMeta \,\synDefEq~
%   &
%     \fNumConstsMeta \mid \trueConst \mid
%     \falseConst \mid 
%     \FADD{\fExprsMeta}{\fExprsMeta} \mid
%     \FSUB{\fExprsMeta}{\fExprsMeta} \mid
%     \FMUL{\fExprsMeta}{\fExprsMeta} \mid
%     \FDIV{\fExprsMeta}{\fExprsMeta} \mid
%     \FEQ{\fExprsMeta}{\fExprsMeta} \mid
%     \FLT{\fExprsMeta}{\fExprsMeta} \mid
%     \FNEG{\fExprsMeta} \mid
%     \FAND{\fExprsMeta}{\fExprsMeta} \mid \\
%   &
%     \IFTHENELSE{\fExprsMeta}{\fExprsMeta}{\fExprsMeta} \mid \\
%   &
%     \NIL \mid \CONCAT{\fExprsMeta}{\fExprsMeta} \mid
%     \LCASE{\fExprsMeta}{\fExprsMeta}{\fExprsMeta} \mid \\ 
%   &
%     \fVarsMeta \mid
%     \FAPP{\fExprsMeta}{\fExprsMeta} \mid
%     \LAM{\fVarsMeta}{\fExprsMeta} \mid 
% %    \LETIN{\fVarsMeta}{\fExprsMeta}{\fExprsMeta} \mid 
%     \LETRECIN{\fVarsMeta}{\LAM{\fVarsMeta'}{\fExprsMeta}}{\fExprsMeta}
% \end{align*}
% \normalsize
%
A program of the eager functional language is an expression.
The syntax for expressions is given in the left part of
\cref{fig:exprs-and-canonical-forms}.
Here, $\fNumConstsMeta$ is a numeral,
$\fVarsMeta$ is a variable,
$\FAPP{\fExprsMeta}{\fExprsMeta'}$
is an application,
$\LAM{\fVarsMeta}{\fExprsMeta}$ is a lambda abstraction, 
$\NIL$ is the empty list,
and $\CONCAT{\exprsMeta_1}{\exprsMeta_2}$ is the list obtained by
prefixing the list $\exprsMeta_2$ with the element $\exprsMeta_1$.
The expression $\LCASE{\fExprsMeta}{\fExprsMeta'}{\fExprsMeta''}$
branches to $\fExprsMeta'$ or $\fExprsMeta''$ depending on whether the
result of $\fExprsMeta$ is the empty list $\NIL$.
The expression
$\LETRECIN{\fVarsMeta}{\LAM{\fVarsMeta'}{\fExprsMeta'}}{\fExprsMeta}$
binds $\fVarsMeta$ to $\LAM{\fVarsMeta'}{\fExprsMeta'}$ in $\fExprsMeta$.
This expression allows $\fVarsMeta$ to be used in $\fExprsMeta'$,
thereby allowing recursion. 

%
% Hence, an expression can be a numeral $\fNumConstsMeta$,
% %
% a Boolean constant,
% %
% an arithmetic expression,
% %
% a Boolean expression,
% %
% a branching expression
% $\IFTHENELSE{\fExprsMeta}{\fExprsMeta_1}{\fExprsMeta_2}$,
% %
% an empty list $\NIL$, 
% %
% a list $\CONCAT{\fExprsMeta}{\fExprsMeta'}$ obtained by prefixing an
% element to a list, 
%
% The expression $\LCASE{\fExprsMeta}{\fExprsMeta'}{\fExprsMeta''}$
% evaluates $\fExprsMeta'$ if $\fExprsMeta$ evaluates to the empty list
% $\NIL$, and evaluates the application of $\fExprsMeta''$ to the head
% and tail of the list resulting from the evaluation of $\fExprsMeta$ if
% this result is a non-empty list.

\vspace*{-1ex}
\paragraph{Semantics. } 
The evaluation of the expressions results in canonical forms
$\cfmsMeta$ as given in the right part of
\cref{fig:exprs-and-canonical-forms}.
A canonical form $\cfmsMeta$ can be
a canonical form for integers ($\intCfmsMeta$),
a canonical form for Boolean values ($\boolCfmsMeta$),
a canonical form for functions ($\funCfmsMeta$), or
a canonical form for lists ($\lstCfmsMeta$).

% \vspace*{-2ex}
% \small
% \begin{align*}
%   \intCfmsMeta \,\synDefEq\,
%   &
%     \ldots \mid -2 \mid -1 \mid 0 \mid 1 \mid 2 \mid \ldots \\
%   \boolCfmsMeta \,\synDefEq\,
%   &
%     \trueConst \mid \falseConst \\
%   \funCfmsMeta \,\synDefEq\,
%   &
%     \LAM{\fVarsMeta}{\fExprsMeta} \\
%   \lstCfmsMeta \,\synDefEq\,
%   &
%     \NIL \mid \CONCAT{\cfmsMeta}{\cfmsMeta} \\
%   \cfmsMeta \,\synDefEq\,
%   &
%     \intCfmsMeta \mid \boolCfmsMeta \mid
%     \funCfmsMeta \mid \lstCfmsMeta 
% \end{align*}
% \normalsize

\begin{figure}[t]
  \small
  \begin{align*}
    ~\\[-5ex]
    &
      \RULE{(\cfmsMeta,\cfmsMeta)}{\,} \\
    &
      \RULE
      {(\fExprsMeta_1\,\op\,\fExprsMeta_2,
      \interpICfm{\op}(\intCfmsMeta_{\!1},\intCfmsMeta_{\!2}))}  
      {(\fExprsMeta_{1}, \intCfmsMeta_{\!1}),
      (\fExprsMeta_2,\intCfmsMeta_{\!2})}
      ~
      \mrm{where}~
      \op\in\{+,-,*,/,=,<\}\\
    &
      \RULE
      {(\lnot \fExprsMeta, \interpBCfm{\lnot}(\boolCfmsMeta))}
      {(\fExprsMeta, \boolCfmsMeta)} \\
    &
      \RULE
      {(\fExprsMeta_1\,\land\,\fExprsMeta_2,
      \interpBCfm{\land}(\boolCfmsMeta_{\!1},\boolCfmsMeta_{\!2}))}
      {(\exprsMeta_1,\boolCfmsMeta_{\!1}),
      (\exprsMeta_2,\boolCfmsMeta_{\!2})}
    \\
    &
      \RULE
      {
      (\IFTHENELSE{\fExprsMeta}{\fExprsMeta'}{\fExprsMeta''}, 
      \cfmsMeta) 
      }
      {
      (\fExprsMeta,\trueConst),
      (\fExprsMeta',\cfmsMeta) 
      }
    \\
    &
      \RULE
      {
      (\IFTHENELSE{\fExprsMeta}{\fExprsMeta'}{\fExprsMeta''},
      \cfmsMeta) 
      }
      {
      (\fExprsMeta,\falseConst),
      (\fExprsMeta'',\cfmsMeta) 
      }
    \\
    % &
    %   \RULE
    %   {(\NIL, \NIL)} 
    %   {}
    % \\
    &
      \RULE
      {(\CONCAT{\fExprsMeta}{\fExprsMeta'},
      \CONCAT{\cfmsMeta}{\cfmsMeta'})}
      {
      (\fExprsMeta, \cfmsMeta),
      (\fExprsMeta',\cfmsMeta') 
      }
    \\
    &
      \RULE
      {
      (\LCASE{\fExprsMeta}{\fExprsMeta'}{\fExprsMeta''},
      \cfmsMeta)
      }
      {
      (\fExprsMeta,\NIL),
      (\fExprsMeta',\cfmsMeta)
      }
    \\
    &
      \RULE
      {
      (\LCASE{\fExprsMeta}{\fExprsMeta'}{\fExprsMeta''},
      \cfmsMeta'') 
      }
      {
      (\fExprsMeta,\CONCAT{\cfmsMeta}{\cfmsMeta'}),
      (\FAPP{\FAPP{\fExprsMeta''}{\cfmsMeta}}{\cfmsMeta'},
      \cfmsMeta'')
      }
      \\
    &
      \RULE
      {
      (\FAPP{\fExprsMeta}{\fExprsMeta'}, \cfmsMeta)
      }
      {
      (\fExprsMeta, \LAM{\fVarsMeta}{\fExprsMeta''}),
      (\fExprsMeta', \cfmsMeta'),
      (\fExprsMeta''[\cfmsMeta'/\fVarsMeta], \cfmsMeta)
      }
    \\
    &
      \RULE
      {
      (\LETRECIN{\fVarsMeta}{\LAM{\fVarsMeta'}{\fExprsMeta'}}{\fExprsMeta},
      \cfmsMeta)
      }
      {
      (\FAPP{(\LAM{\fVarsMeta}{\fExprsMeta})}
      {(\LAM{\fVarsMeta'}
      {\LETRECIN{\fVarsMeta}{\LAM{\fVarsMeta'}{\fExprsMeta'}}{\fExprsMeta'}})},
      \cfmsMeta 
      )
      }
      ~~\mrm{if}~
      \fVarsMeta'\neq\fVarsMeta
  \end{align*}
  \normalsize
  \vspace*{-4.5ex}
  \caption{The semantic rules for the eager functional language}
    \vspace*{-3.5ex}
  \label{fig:sem-rules-eager-fun-lang}
\end{figure}

We formulate the big-step semantics of the eager functional language
by defining the predicate $\mi{rule}$ as in
\cref{fig:sem-rules-eager-fun-lang}.
In this figure, the fist rule says that an expression that is a
canonical form evaluates to itself.
%
% These expressions include the numerals (that are identified with the
% non-negative integer canonical forms), the Boolean constants, the
% expression $\LAM{\fVarsMeta}{\fExprsMeta}$, the empty list $\NIL$, and
% the list expression obtained by concatenating canonical forms.
%
In the next few rules, $\interpICfm{\op}$ is used to represent the
interpretation of the binary operation $\op$ on integer canonical
forms.
In addition, $\interpBCfm{\lnot}$ and $\interpBCfm{\land}$ represent
the interpretations of the logical operations $\lnot$ and $\land$,
respectively, on Boolean canonical forms. 
The expression $\fSubst{\exprsMeta}{\varsMeta}{\cfmsMeta}$ represents
substitutions of $\cfmsMeta$ for $\varsMeta$ in $\exprsMeta$.
This expression is defined in \cref{asec:subst}.
In the rule for $\mathsf{letrec}$, the evaluation of the expression
$\LETRECIN{\fVarsMeta}{\LAM{\fVarsMeta'}\fExprsMeta'}{\fExprsMeta}$ is
turned into the evaluation of
$\FAPP{(\LAM{\fVarsMeta}{\fExprsMeta})} {(\LAM{\fVarsMeta'}
  {\LETRECIN{\fVarsMeta}{\LAM{\fVarsMeta'}{\fExprsMeta'}}{\fExprsMeta'}})}$.
Intuitively, when $\fExprsMeta$ is evaluated, $\fVarsMeta$ can be
applied as
$(\LAM{\fVarsMeta'}
{\LETRECIN{\fVarsMeta}{\LAM{\fVarsMeta'}{\fExprsMeta'}}{\fExprsMeta'}})$.
When the second $\fExprsMeta'$ in
$(\LAM{\fVarsMeta'}
{\LETRECIN{\fVarsMeta}{\LAM{\fVarsMeta'}{\fExprsMeta'}}{\fExprsMeta'}})$
is evaluated, the name $\fVarsMeta$ is still available and bound to a
$\lambda$-abstraction.

\vspace*{-2ex}
\subsubsection*{List-Merging Program and its Verification}

% \begin{figure}[t]
%   \small
%   \begin{align*}
%     \mergeExpr(\lstCfmsMeta_{\!1},\lstCfmsMeta_{\!2})
%     \,\defEq \, 
%     &
%       \mathsf{letrec}~
%       \mergeVar =
%       (\lambda \xFVar.\lambda \xpFVar.\lcaseExpr)~
%     % \\[-.5ex]
%     % &
%       \mathsf{in}~
%       \FAPP{\FAPP{\mergeVar}{\lstCfmsMeta_{\!1}}}{\lstCfmsMeta_{\!2}} 
%     \\
%     \lcaseExpr
%     \,\defEq\,
%     &
%       \mathsf{listcase}~\xFVar~\mathsf{of}~
%       (\xpFVar, \lambda\iFVar. \lambda \rFVar.
%       \mathsf{listcase}~\xpFVar~\mathsf{of}~
%       (\xFVar, \lambda \ipFVar. \lambda \rpFVar. 
%       \ifExpr))
%     \\
%     \ifExpr \,\defEq\,
%     &
%       \IFTHENELSE
%       {\iFVar\le \ipFVar}
%       {\CONCAT{\iFVar}{\FAPP{\FAPP{\mergeVar}{\rFVar}}{\xpFVar}}}
%       {\CONCAT{\iFVar'}{\FAPP{\FAPP{\mergeVar}{\xFVar}}{\rpFVar}}} 
%   \end{align*}
%   \normalsize
%   \caption{The program that merges sorted lists}
%   \label{fig:prog-merge-list} 
% \end{figure}

The program $\mergeExpr(\lstCfmsMeta_{\!1},\lstCfmsMeta_{\!2})$ below
merges two sorted lists into a single sorted list.
More concretely, the variable $\mergeFId$ is bound to the expression
$\lambda\xFVar.\lambda\xpFVar.\lcaseExpr$ that destructs the lists
that are bound to $\xFVar$ and $\xpFVar$, respectively.
%$\lstCfmsMeta_1$ and $\lstCfmsMeta_2$, which are given as arguments.
%
In case one of the lists is empty, the result of the merger is the
other list.
Otherwise, the result of the merger is obtained by prefixing the
smaller head element of the two given lists over the merging result of
the remaining parts of the lists.

\vspace*{-3ex}
\small
  \begin{align*}
    \mergeExpr(\lstCfmsMeta_{\!1},\lstCfmsMeta_{\!2})
    \,\defEq \, 
    &
      \mathsf{letrec}~
      \mergeVar =
      (\lambda \xFVar.\lambda \xpFVar.\lcaseExpr)~
    % \\[-.5ex]
    % &
      \mathsf{in}~
      \FAPP{\FAPP{\mergeVar}{\lstCfmsMeta_{\!1}}}{\lstCfmsMeta_{\!2}} 
    \\
    \lcaseExpr
    \,\defEq\,
    &
      \mathsf{listcase}~\xFVar~\mathsf{of}~
      (\xpFVar, \lambda\iFVar. \lambda \rFVar.
      \mathsf{listcase}~\xpFVar~\mathsf{of}~
      (\xFVar, \lambda \ipFVar. \lambda \rpFVar. 
      \ifExpr))
    \\
    \ifExpr \,\defEq\,
    &
      \IFTHENELSE
      {\iFVar\le \ipFVar}
      {\CONCAT{\iFVar}{\FAPP{\FAPP{\mergeVar}{\rFVar}}{\xpFVar}}}
      {\CONCAT{\iFVar'}{\FAPP{\FAPP{\mergeVar}{\xFVar}}{\rpFVar}}} 
  \end{align*}
  \normalsize
\vspace*{-3ex}  

To develop a specification for the list-merging program, we define
a piece of auxiliary notation.
We write $\listOfLstCfm{\lstCfmsMeta}$ for the mathematical list of
integers represented by the canonical form $\lstCfmsMeta$ for lists.
Formally, we define
$\listOfLstCfm{\NIL}\defEq \emptyList$,
$\listOfLstCfm{\CONCAT{\intCfmsMeta}{\lstCfmsMeta}}$
$\defEq
{\intCfmsMeta}\lstPrepend{\listsOfIntsMeta}$ if
$\listsOfIntsMeta= \listOfLstCfm{\lstCfmsMeta}\,\land\,
\listsOfIntsMeta \in \ints^{*}$, and 
$\listOfLstCfm{\lstCfmsMeta}\defEq\bot$ otherwise.
     
% \vspace*{-2ex}
% \small
% \begin{align*}
%   \listOfLstCfm{\NIL} \,\defEq\,
%   &
%   \emptyList \\
%   \listOfLstCfm{\CONCAT{\intCfmsMeta}{\lstCfmsMeta}}
%   \,\defEq\,
%   &
%   \begin{cases}
%     {\intCfmsMeta}\lstPrepend{\listsOfIntsMeta} 
%     & \mrm{if}~\listsOfIntsMeta=
%     \listOfLstCfm{\lstCfmsMeta}\,\land\,
%     \listsOfIntsMeta \in \ints^{*} \\
%     \bot
%     &
%     \mrm{otherwise}
%   \end{cases}
%   \\
%   \listOfLstCfm{\lstCfmsMeta}\,\defEq\,
%   &
%     \bot
%     \quad 
%     \mrm{if}~\lstCfmsMeta~\mrm{is~not~of~the~above~forms}
% \end{align*}
% \normalsize

We devise the a specification for the list-merging program, 
$\specMGList$. 
Using the function $\mi{occ}$ and the predicate $\mi{sorted}$
introduced in \cref{sec:ext-while}, we specify the expression
$\mergeExpr(\lstCfmsMeta_{\!1},\lstCfmsMeta_{\!2})$ as

\vspace*{-2ex}
\small
\begin{align*}
  &
    \begin{aligned}
      &
      \specMGList(\mergeExpr(\lstCfmsMeta_{\!1},\lstCfmsMeta_{\!2}))
      \,\defEq\, \\
      &\qquad 
      \{ \lstCfmsMeta \mid \exists \listsOfIntsMeta\in \ints^*\!:
      \listsOfIntsMeta = \listOfLstCfm{\lstCfmsMeta}\,\land
      \occ{\listsOfIntsMeta}=
      \occAdd{\occ{\listsOfIntsMeta_1}}{\occ{\listsOfIntsMeta_2}}\,\land
      \, \sorted{\listsOfIntsMeta} \}
    \end{aligned}
  \\
  &
    \mrm{if}~
    \listsOfIntsMeta_1\in \ints^*\,\land\,
    \listsOfIntsMeta_2\in \ints^*\,\land\,
    \sorted{\listsOfIntsMeta_1}\,\land\,
    \sorted{\listsOfIntsMeta_2}
  \\
  &\mrm{where}~
    \listsOfIntsMeta_1=\listOfLstCfm{\lstCfmsMeta_{\!1}}\,\land\,
    \listsOfIntsMeta_2=\listOfLstCfm{\lstCfmsMeta_{\!2}}
    \\[-4ex]
\end{align*}
\normalsize
This specification says that given list canonical forms
$\lstCfmsMeta_{\!1}$ and $\lstCfmsMeta_{\!2}$ that are both sorted in
ascending order, the result of executing
$\mergeExpr(\lstCfmsMeta_{\!1},\lstCfmsMeta_{\!2})$ is a list canonical form
$\lstCfmsMeta$.
The list canonical form $\lstCfmsMeta$ contains the elements as
contained in either $\lstCfmsMeta_{\!1}$ or $\lstCfmsMeta_{\!2}$.
Furthermore, the list canonical form $\lstCfmsMeta$ is sorted in
ascending order.

To support the verification of the specification for
$\mergeExpr(\lstCfmsMeta_{\!1},\lstCfmsMeta_{\!2})$, we specify an
unfolded form of this expression. 
The execution of this unfolded form either terminates directly, or
gives the same form again.

\vspace*{-2ex}
\small
\begin{align*}
  &
    \begin{aligned}
      &
      \specMGList(\,
      (\LAM{\xFVar}
      {\LETRECIN{\mergeFId}
        {\LAM{\xFVar}{\LAM{\xpFVar}{\lcaseExpr}}}
        {\LAM{\xpFVar}{\lcaseExpr}}})
      \,\, \lstCfmsMeta_{\!1}\, \lstCfmsMeta_{\!2}\,) 
      \, \defEq\,  \\
      & \qquad
      \{
      \lstCfmsMeta \mid
      \exists \listsOfIntsMeta\in \ints^*\!:
      \listsOfIntsMeta = \listOfLstCfm{\lstCfmsMeta}\,\land\,
      \occ{\listsOfIntsMeta}=
      \occAdd{\occ{\lstCfmsMeta_{\!1}}}{\occ{\lstCfmsMeta_{\!2}}}\,\land \, 
      \sorted{\listsOfIntsMeta} 
      \}
    \end{aligned}
  \\
  &
    \mrm{if}~
    \listsOfIntsMeta_1\in \ints^*\,\land\,
    \listsOfIntsMeta_2\in \ints^*\,\land\,
    \sorted{\listsOfIntsMeta_1}\,\land\,
    \sorted{\listsOfIntsMeta_2}
  \\
  &\mrm{where}~
    \listsOfIntsMeta_1=\listOfLstCfm{\lstCfmsMeta_{\!1}}\,\land\,
    \listsOfIntsMeta_2=\listOfLstCfm{\lstCfmsMeta_{\!2}}
\end{align*}
\normalsize
This specification reflects that the unfolded expression
$(\lambda \xFVar. \mathsf{letrec}~\mergeFId\,=\,\lambda \xFVar.
\lambda \xpFVar.$
$\lcaseExpr\,\mathsf{in}\, \lambda \xpFVar.\lcaseExpr) \,\,
\lstCfmsMeta_{\!1}\, \lstCfmsMeta_{\!2}$
% \[(\LAM{\xFVar} {\LETRECIN{\mergeFId}
%     {\LAM{\xFVar}{\LAM{\xpFVar}{\lcaseExpr}}}
%     {\LAM{\xpFVar}{\lcaseExpr}}}) \,\, \lstCfmsMeta_{\!1}\,
%   \lstCfmsMeta_{\!2}\]
delivers analogous guarantees to those
delivered by the original expression
$\mergeExpr(\lstCfmsMeta_{\!1},\lstCfmsMeta_{\!2})$.

With the technique of \cref{sec:method}, the validity of
% the specification
$\specMGList$ can be established. 
\begin{theorem}\label{thm:fun_correct}
  It holds that $\specValid{\specMGList}$.  
\end{theorem}
%
% The proof
% % of this theorem
% is performed using \cref{thm:soundness}.
%
With the help of \cref{thm:soundness}, the proof requires no induction
for reasoning about the recursive applications of the function bound
to $\mergeFId$.
This proof boils down to symbolic execution with the help of a few
auxiliary lemmas about substitution and evaluation related to
canonical forms.

\begin{remark}
  It might appear that the auxiliary information needed for the
  verification of the list-merging program should be for expressions
  of the form $\mergeFId~\wildcard~\wildcard$. 
  % -- because
  % these are the recursive applications of $\mergeFId$ in 
  % $\mergeExpr(\lstCfmsMeta_{\!1},\lstCfmsMeta_{\!2})$.
  %
  However, these expressions 
  % expressions of the form $\mergeFId~\wildcard~\wildcard$ 
  cannot be evaluated, because information about the actual function
  bound to $\mergeFId$ is missing.
  The form that recurs in the evaluation of
  $\mergeExpr(\lstCfmsMeta_{\!1},\lstCfmsMeta_{\!2})$ is actually
  $(\LAM{\xFVar} {\LETRECIN{\mergeFId}
    {\LAM{\xFVar}{\LAM{\xpFVar}{\lcaseExpr}}}
    {\LAM{\xpFVar}{\lcaseExpr}}}) \,\, \wildcard \,\, \wildcard$.
\end{remark}

% The proof of this theorem is performed using \cref{thm:soundness}. 
% %
% The use of \cref{thm:soundness} factors out the inductions needed for
% reasoning about the recursive applications of the function $\mergeFId$. 
% %
% Thus, the proof consists mainly of symbolic execution and the
% application of a few auxiliary results about substitution and
% evaluation related to canonical forms.

% \small
% \begin{align*}
%   &
%   \begin{aligned}
%     \specMGList(\FAPP{\FAPP{\mergeVar}{\lstCfmsMeta_1}}{\lstCfmsMeta_2})
%     \,\defEq\, & \{ \lstCfmsMeta \mid \occ{\lstCfmsMeta}=
%     \occAdd{\occ{\lstCfmsMeta_1}}{\occ{\lstCfmsMeta_2}}\,\land \,
%     \sorted{\lstCfmsMeta} \}
%   \end{aligned}
%   \\
%   &
%     \mrm{if}~
%     \exists \listsOfIntsMeta_1,\listsOfIntsMeta_2
%     \in \ints^{*}:
%     \listsOfIntsMeta_1=\listOfLstCfm{\lstCfmsMeta_1}\,\land\,
%     \listsOfIntsMeta_2=\listOfLstCfm{\lstCfmsMeta_2}\,\land\,
%     \sorted{\listsOfIntsMeta_1}\,\land\,
%     \sorted{\listsOfIntsMeta_2} 
% \end{align*}
% \normalsize

% \subsubsection*{Proof Outline for List-Merging Program}

%%% Local Variables:
%%% mode: latex
%%% TeX-master: "main"
%%% End:

%%% Local Variables:
%%% mode: latex
%%% TeX-master: "main"
%%% End:

\vspace*{-2.5ex}
\section{On Completeness of the Technique} \label{sec:completeness}
\vspace*{-1ex}

%We next consider the completeness of our proof system.
%
It is untrue that any valid specification can be verified.
Intuitively, a specification $\specsMeta$ that is valid but missing
the necessary auxiliary information such as loop invariants
% for loops
might not be verifiable.
However, we show that there is always a more informative specification
$\specsMeta'$ than $\specsMeta$ that is verifiable.

Formally, a specification $\specsMeta_2$ is \emph{at least as
  informative as} a specification $\specsMeta_1$, as denoted by
$\fSpecLe{\specsMeta_1}{\specsMeta_2}$,
if for each configurations $\configsMeta$, it holds that
$\specsMeta_1(\configsMeta)\supseteq \specsMeta_2(\configsMeta)$. 
% if for each configuration $\configsMeta$ such that
% $\specsMeta_1(\configsMeta)=\setsOfRConfigsMeta\neq\bot$, there exists
% some $\setsOfRConfigsMeta'$ such that
% $\setsOfRConfigsMeta'\subseteq \setsOfRConfigsMeta$ and
% $\specsMeta_2(\configsMeta)=\setsOfRConfigsMeta'$.

The lemma below says the specification mapping each configuration to
the set of all the semantically derivable result configurations can be
verified.

%  \vspace*{-1ex}
\begin{lemma}\label{lem:complete}
  Let
  $\specsMeta_\star\defEq \lambda \configsMeta. \{\rConfigsMeta\mid
  \DERIV{\configsMeta, \rConfigsMeta}\}$.
  Then, $\verifySpec{\specsMeta_\star}$ can be established. 
\end{lemma}
  \vspace*{-2ex}

\begin{proof}
  See \cref{asec:proof} and the mechanization in Coq. \qed
\end{proof}

The following theorem says that for each valid specification
$\specsMeta$, there is a valid specification that is at least as
informative as $\specsMeta$, and that can be verified. 

\begin{theorem}[Relative Completeness]
  For each valid specification $\specsMeta$, there exists a
  specification $\specsMeta'$ such that
  $\fSpecLe{\specsMeta}{\specsMeta'}$, and
  %holds, and
  $\verifySpec{\specsMeta'}$ can be established.
\end{theorem}
\renewcommand*{\proofname}{Proof outline}
\begin{proof}
  It is not difficult to show that the specification
  $\specsMeta_\star$ is at least as informative as any valid
  specification. 
  % It is not difficult to show that the specification $\specsMeta_\star$
  % of \cref{lem:complete} is valid.
  % %
  % It is also not difficult to show that $\specsMeta_\star$ is at least
  % as informative as any valid specification.
  %
  Thus, the conclusion of the theorem follows from
  \cref{lem:complete}.  \qed
\end{proof}
\renewcommand*{\proofname}{Proof}

%Intuitively, 
If the program contained in a configuration exhibits
%admits
only bounded behavior, then the corresponding result configuration can
be obtained through symbolic execution. % using the semantics.
Hence, it is not necessary that a verifiable specification should
cover these configurations.
In an informal sense, this argument supports that for a specification
to be verified, it is only necessary to provide auxiliary information
about constructs such as loops and recursive function calls in the
specification.

%%% Local Variables:
%%% mode: latex
%%% TeX-master: "main"
%%% End:

\vspace*{-2ex}
\section{Discussion}\label{sec:discussion}
\vspace*{-1ex}

% \subsubsection*{Big-step Semantics of Programming Languages}
% %
% There are many existing formalizations of big-step semantics for
% practical programming languages.
% %
% Examples include the big-step semantics for the two sub-languages,
% Clight~\cite{BlazyLeroy2009} and Simpl~\cite{Schirmer2008}, of C,
% %
% the big-step semantics for the two sublanguages,
% Java$_{\mrm{light}}$~\cite{NipkowOheimb1998} and
% Jinja~\cite{KleinNipkow2005}, of Java,
% %
% the big-step semantics for the Solidity-like blockchain language
% Lolisa~\cite{YangLei2018}, etc.
% %
% The extensive formalization and usage of big-step semantics 
% for practically relevant purposes indicates the potential values
% of our reasoning technique that is based on big-step semantics. 

\subsubsection*{Reuse of Existing Formalization of Semantics}
With the verification techniques based on small-step
semantics~\cite{Moore2003,MoorePenaRosu2018}, it is not difficult to
obtain a verification infrastructure by reusing an existing
formalization of semantics.
This is because a small-step semantics readily provides a step relation 
that can be used to interface with the verification framework. 
In comparison, we have only shown that our language-independent
verification technique can be applied after the big-step semantics of
the target language is formalized via a predicate that explicitly
captures the premises and conclusions of the semantic rules.
Although we have demonstrated in \cref{sec:verification} with 
different types of languages that the big-step semantics formulated 
using this predicate closely resemble their classical formulation, 
it is desirable if a higher level of reusability can be enabled. 
A potential solution is to construct a program that automatically
transforms a formalization of big-step semantics into a formulation
with the $\mi{rule}$ predicate.
%
% Another possibility is to perform the development in a similar spirit 
% to the present work, but assuming a semantic function, 
% rather than a relation. 
% %
% In this way, the reuse of semantics formalized using the 
% functional big-step approach~\cite{OwensMyreenKumarThiemann2016} 
% could be more easily achievable. 
% %
% The further exploration of some of the ideas discussed in the above 
% is the subject of our ongoing research.

\vspace*{-2ex}
\subsubsection*{Integration of Techniques Dealing with other Aspects}
The purpose of the present work is not to simplify the overall task of
deductive program verification beyond achievable by existing
techniques. 
% The purpose of the present work is not to simplify the overall task of
% deductive program verification beyond what is achieved by existing
% techniques.
%
Instead, the focus has been the simplification of the verification for
constructs causing potentially unbounded behavior, 
%such as loops and recursive calls,
based on a common model of big-step operational semantics.
% in a feature-oriented, rather than language-oriented fashion.
%
To construct a full-fledged language-independent program verifier in a
proof assistant, effective treatment of other aspects of deductive
program verification (e.g., memory layout, mathematical reasoning for
specific problem domains, etc.), as well as further techniques for
increasing the level of automation, is required.
Existing work in program logics, program verifiers, and theory
libraries in proof assistants are expected to be a crucial source of
inspiration as well as concrete technical components in dealing with
the remaining aspects of the verification.
%
% A substantial amount of effort will be needed to build such a program
% verifier.
%
%
% techniques for effectively dealing with other aspects of
% deductive program verification (e.g., memory layout, mathematical
% reasoning for specific problem domains, etc.), as well as further
% tactics for increasing the level of automation, should be integrated.
%
% Nevertheless, we provide a basis for such a verifier through the
% present work.

%\subsubsection*{Inference of the Specification} 

%%% Local Variables: 
%%% mode: latex
%%% TeX-master: "main"
%%% End:

\vspace*{-1ex}
\section{Conclusion} \label{sec:conclusion}
\vspace*{-.5ex}

To tackle the problem caused by the proliferation of programming
languages in deductive program verification, we provide a technique to
address the cross-cutting concern of reasoning about language features
causing unbounded behavior, including loops and recursive function
calls.
The technique can be applied to any programming language, as long as
the big-step operational semantics of the language is formulated with
an explicit characterization of the premises and conclusions of
inference rules.
The user of this technique need not set up inductions for the loops
and recursive calls in performing a program proof, but formulates
invariants and function contracts in a uniform style with a
specification, and perform symbolic execution of the program with the
help of this specification. 
The technique admits succinct, inductive arguments for soundness and
relative completeness that are verified in the Coq proof assistant
along with other formal claims~\cite{Li2021}.
It has been illustrated with verification examples targeting languages
of different paradigms.
It provides a basis for a language-independent tool for program
verification based on big-step operational semantics in proof
assistants.

% Potential directions for future work include the improvement of the
% level of automation (e.g., through specification inference), and the
% integration of effective technical components for dealing with
% additional aspects of deductive program verification such as memory
% layout and domain-specific mathematical reasoning.

\newpage

%%% Local Variables:
%%% mode: latex
%%% TeX-master: "main"
%%% End:

\subsubsection*{Acknowledgment} 
The research was supported by the National Natural Science Foundation of
China (61876111, 61877040, 62002246), the general project numbered
KM202010028010 of Beijing Municipal Education Commission, and the Open
Project CARCH201920 of State Key Laboratory of Computer Architecture,
Institute of Computing Technology, Chinese Academy of Sciences.

\bibliographystyle{abbrv}
\bibliography{references}

\begin{thebibliography}{10}

\bibitem{Coq}
The {Coq} proof assistant.
\newblock https://coq.inria.fr/.

\bibitem{Michelson}
Michelson -- the language of {T}ezos.
\newblock https://www.michelson.org/.

\bibitem{Move}
The move language.
\newblock
  https://developers.libra-china.org/docs/crates/move-language/index.html.

\bibitem{Schirmer2008}
A sequential imperative programming language -- syntax, semantics, {H}oare
  logics and verification environment.
\newblock https://www.isa-afp.org/entries/Simpl.html.

\bibitem{Solidity}
Solidity.
\newblock https://docs.soliditylang.org/en/v0.8.0/.

\bibitem{VCC}
{VCC}: A verifier for concurrent {C}.
\newblock
  https://www.microsoft.com/en-us/research/project/vcc-a-verifier-for-concurrent-c/.

\bibitem{Yul}
Yul.
\newblock https://docs.soliditylang.org/en/v0.8.0/yul.html.

\bibitem{Li2021}
Formalization of the verification technique in {C}oq.
\newblock https://github.com/lixm/ind-verify/tree/master, 2021.

\bibitem{AhrendtBBHSU2016}
W.~Ahrendt, B.~Beckert, R.~Bubel, R.~H{\"{a}}hnle, P.~H. Schmitt, and
  M.~Ulbrich, editors.
\newblock {\em Deductive Software Verification - The KeY Book - From Theory to
  Practice}, volume 10001 of {\em Lecture Notes in Computer Science}.
\newblock Springer, 2016.

\bibitem{Appel2011}
A.~W. Appel.
\newblock Verified software toolchain - (invited talk).
\newblock In {\em Proceedings of 20th European Symposium on Programming
  ({ESOP})}, pages 1--17, 2011.

\bibitem{BlazyLeroy2009}
S.~Blazy and X.~Leroy.
\newblock Mechanized semantics for the {C}light subset of the {C} language.
\newblock {\em Journal of Automated Reasoning}, 43(3):263--288, 2009.

\bibitem{BodinGardnerJensenSchmitt2019}
M.~Bodin, P.~Gardner, T.~P. Jensen, and A.~Schmitt.
\newblock Skeletal semantics and their interpretations.
\newblock {\em Proc. {ACM} Program. Lang.}, 3({POPL}):44:1--44:31, 2019.

\bibitem{BodinJensenSchmitt2015}
M.~Bodin, T.~P. Jensen, and A.~Schmitt.
\newblock Certified abstract interpretation with pretty-big-step semantics.
\newblock In {\em Proceedings of the 2015 Conference on Certified Programs and
  Proofs ({CPP})}, pages 29--40, 2015.

\bibitem{ClementDespeyrouxDespeyrouxKahn86}
D.~Cl{\'{e}}ment, J.~Despeyroux, T.~Despeyroux, and G.~Kahn.
\newblock A simple applicative language: {M}ini-{ML}.
\newblock In {\em Proceedings of the 1986 {ACM} Conference on {LISP} and
  Functional Programming ({LFP})}, pages 13--27, 1986.

\bibitem{CousotCousot1977}
P.~Cousot and R.~Cousot.
\newblock Abstract interpretation: {A} unified lattice model for static
  analysis of programs by construction or approximation of fixpoints.
\newblock In {\em Fourth {ACM} Symposium on Principles of Programming Languages
  (POPL)}, pages 238--252, 1977.

\bibitem{Hirai2017}
Y.~Hirai.
\newblock Defining the ethereum virtual machine for interactive theorem
  provers.
\newblock In {\em Financial Cryptography and Data Security - {FC} 2017
  International Workshops}, pages 520--535, 2017.

\bibitem{Hoare1969}
C.~A.~R. Hoare.
\newblock An axiomatic basis for computer programming.
\newblock {\em Communications of the {ACM}}, 12(10):576--580, 1969.

\bibitem{JungKJBBD18}
R.~Jung, R.~Krebbers, J.~Jourdan, A.~Bizjak, L.~Birkedal, and D.~Dreyer.
\newblock Iris from the ground up: {A} modular foundation for higher-order
  concurrent separation logic.
\newblock {\em Journal of Functional Programming}, 28:e20, 2018.

\bibitem{Kahn1987}
G.~Kahn.
\newblock Natural semantics.
\newblock In {\em Proceedings of 4th Annual Symposium on Theoretical Aspects of
  Computer Science ({STACS})}, pages 22--39, 1987.

\bibitem{KleinNipkow2005}
G.~Klein and T.~Nipkow.
\newblock Jinja is not {J}ava.
\newblock {\em Arch. Formal Proofs}, 2005.

\bibitem{McCarthy1962}
J.~McCarthy.
\newblock Towards a mathematical science of computation.
\newblock In {\em Information Processing, Proceedings of the 2nd {IFIP}
  Congress}, pages 21--28, 1962.

\bibitem{MoorePenaRosu2018}
B.~M. Moore, L.~Pe{\~{n}}a, and G.~Rosu.
\newblock Program verification by coinduction.
\newblock In {\em Programming Languages and Systems - Proceedings of 27th
  European Symposium on Programming ({ESOP})}, pages 589--618, 2018.

\bibitem{Moore2003}
J.~S. Moore.
\newblock Inductive assertions and operational semantics.
\newblock In {\em Correct Hardware Design and Verification Methods, 12th {IFIP}
  {WG} 10.5 Advanced Research Working Conference ({CHARME})}, pages 289--303,
  2003.

\bibitem{NielsonNielson2007}
H.~R. Nielson and F.~Nielson.
\newblock {\em Semantics with Applications: An Appetizer}.
\newblock Undergraduate Topics in Computer Science. Springer, 2007.

\bibitem{NipkowOheimb1998}
T.~Nipkow and D.~von Oheimb.
\newblock Java\emph{\({}_{\mbox{light}}\)} is type-safe - definitely.
\newblock In {\em Proceedings of the 25th {ACM} {SIGPLAN-SIGACT} Symposium on
  Principles of Programming Languages ({POPL})}, pages 161--170, 1998.

\bibitem{Plotkin1981}
G.~D. Plotkin.
\newblock A structural approach to operational semantics.
\newblock Lecture notes, DAIMI FN-19, 1981.

\bibitem{Reynolds1998}
J.~C. Reynolds.
\newblock {\em Theories of programming languages}.
\newblock Cambridge University Press, 1998.

\bibitem{Reynolds2002}
J.~C. Reynolds.
\newblock Separation logic: {A} logic for shared mutable data structures.
\newblock In {\em Proceeding of 17th {IEEE} Symposium on Logic in Computer
  Science ({LICS})}, pages 55--74, 2002.

\bibitem{Schmidt1995}
D.~A. Schmidt.
\newblock Natural-semantics-based abstract interpretation (preliminary
  version).
\newblock In {\em Proceedings of Second International Symposium on Static
  Analysis ({SAS})}, pages 1--18, 1995.

\bibitem{SergeyNagarajJohannsenTrunovHao2019}
I.~Sergey, V.~Nagaraj, J.~Johannsen, A.~Kumar, A.~Trunov, and K.~C.~G. Hao.
\newblock Safer smart contract programming with {Scilla}.
\newblock {\em Proc. {ACM} Program. Lang.}, 3({OOPSLA}):185:1--185:30, 2019.

\bibitem{SewellMyreenKlein2013}
T.~A.~L. Sewell, M.~O. Myreen, and G.~Klein.
\newblock Translation validation for a verified {OS} kernel.
\newblock In {\em {ACM} {SIGPLAN} Conference on Programming Language Design and
  Implementation ({PLDI})}, pages 471--482, 2013.

\bibitem{StefanescuParkYuwenLiRosu2016}
A.~Stefanescu, D.~Park, S.~Yuwen, Y.~Li, and G.~Rosu.
\newblock Semantics-based program verifiers for all languages.
\newblock In {\em 2016 ACM SIGPLAN International Conference on Object-Oriented
  Programming, Systems, Languages, and Applications ({OOPSLA})}, pages 74--91,
  2016.

\bibitem{Wood2017}
G.~Wood.
\newblock Ethereum: A secure decentralised generlised transaction ledger.
\newblock https://gavwood.com/paper.pdf.

\bibitem{YangLei2018}
Z.~Yang and H.~Lei.
\newblock Lolisa: Formal syntax and semantics for a subset of the {S}olidity
  programming language.
\newblock {\em CoRR}, abs/1803.09885, 2018.

\end{thebibliography}

\newpage 

\appendix

\section{Supplementary Material for \cref{sec:ext-while}} \label{asec:supp-ext-while}

% \caption{Evaluation of Arithmetic Expressions in the Extended While Language}
% \label{fig:eval-aexp-ext-while}

\subsubsection*{Definitions for the Extended While Language}{~}\\

We define the evaluation function $\mathcal{A}$ for arithmetic
expressions below. 
\small
\begin{align*}
  \ewhAEval{\numeralsMeta}\ewhStatesMeta
  ~\defEq~& \ewhNEval{\numeralsMeta} \\
  \ewhAEval{\varsMeta}\ewhStatesMeta
  ~\defEq~& \ewhStatesMeta(\varsMeta) \\
  \ewhAEval{\arrsMeta}\ewhStatesMeta
  ~\defEq~& \ewhStatesMeta(\arrsMeta) \\
  \ewhAEval{\arrRef{\arrsMeta}{\aExpsMeta_1}}\ewhStatesMeta
  ~\defEq~
          &
            \begin{cases}
              \ewhStatesMeta(\locsMeta+\intsMeta)
              & \mrm{if}~ \exists \locsMeta, \intsMeta: 
              \locsMeta = \ewhAEval{\arrsMeta} 
              \ewhStatesMeta\,\land \,
              \intsMeta = \ewhAEval{\aExpsMeta_1}\ewhStatesMeta \ge 0\,
              \land \, \locsMeta+\intsMeta < \locBoundsMeta \\
              \bot & \mrm{otherwise}
            \end{cases} \\
  \ewhAEval{\aExpsMeta_1\,\aop\,\aExpsMeta_2}\ewhStatesMeta
  ~\defEq~&
            \interp{\aop}
            (\ewhAEval{\aExpsMeta_1}\ewhStatesMeta,
            \ewhAEval{\aExpsMeta_2}\ewhStatesMeta)
\end{align*}
\normalsize
The evaluation of an array identifier $\arrsMeta$ in the state
$\ewhStatesMeta$ yields the starting location of the array identified
by $\ewhStatesMeta$.
The evaluation of the array element expression
$\arrRef{\aExpsMeta_1}{\aExpsMeta_2}$ is the value at the location
$\locsMeta+\intsMeta$, where $\locsMeta$ is the starting location of
the array, and $\intsMeta$ is the index of the array element, if
$\intsMeta$ is non-negative, and the location is within the range
bounded by the next fresh location for arrays.
Otherwise, the evaluation of this expression yields the undefined
value $\bot$.

We define the evaluation function $\mathcal{B}$ for Boolean
expressions of the extended While language as

\vspace*{-2ex}
\small
\begin{align*}
  \ewhBEval{\trueLit}\ewhStatesMeta 
  ~\defEq~& \TT \\
  \ewhBEval{\falseLit}\ewhStatesMeta
  ~\defEq~& \FF \\
  \ewhBEval{\aExpsMeta_1\,\cop\,\aExpsMeta_2}
  \ewhStatesMeta
  ~\defEq~&
            \interp{\cop}(
            \ewhAEval{\aExpsMeta_1}\ewhStatesMeta
            ,
            \ewhAEval{\aExpsMeta_2}\ewhStatesMeta 
            )
  \\
  \ewhBEval{\andExpr{\bExpsMeta_1}{\bExpsMeta_2}}
  \ewhStatesMeta
  ~\defEq~&
            \begin{cases}
              \TT & \mrm{if}~
              \ewhBEval{\bExpsMeta_1}\ewhStatesMeta
              =\TT \land 
              \ewhBEval{\bExpsMeta_2}\ewhStatesMeta
              =\TT \\
              \FF & \mrm{otherwise} 
            \end{cases}
  \\
  \ewhBEval{\notExpr{\bExpsMeta_1}}\ewhStatesMeta
  ~\defEq~&
            \lnot
            \ewhBEval{\bExpsMeta_1}\ewhStatesMeta 
\end{align*}
\normalsize

We denote by
$\callIniSt{\storesMeta}{\listsOfNamesMeta}{\listsOfValsMeta}{\listsOfVarsMeta}$
the initial store for a function with parameters
$\listsOfNamesMeta=\LST{\pdsMeta_1,\dots,\pdsMeta_m}$, argument values
$\listsOfValsMeta=\LST{\valsMeta_1,\dots,\valsMeta_m}$, and return
variables $\listsOfValsMeta=\LST{\varsMeta_1,\dots,\varsMeta_n}$, that
is invoked when the store is $\storesMeta$. 
Formally,
$\callIniSt{\storesMeta}{\listsOfNamesMeta}{\listsOfValsMeta}{\listsOfVarsMeta}$
is the store $\storesMeta'$ such that for each
$\pdsMeta\in \vars\cup\arrs$, it holds that
\[
  \storesMeta'(\pdsMeta)=
  \begin{cases}
    \valsMeta_i & \mrm{if}~i\in \{1,\dots,m\}\,\land\,\pdsMeta=\pdsMeta_i \\
    0 & \mrm{if}~\exists i\in \{1,\dots,n\}:\pdsMeta=\varsMeta_i \\
    \bot & \mrm{otherwise}
  \end{cases}
\]
and for each $\locsMeta\in\nats$, it holds that 
$\storesMeta'(\locsMeta)=\storesMeta(\locsMeta)$.

We denote by
$\callFinSt{\storesMeta}{\storesMeta'}{\listsOfVarsMeta}{\listsOfVarsMeta'}$ the
final store resulting from a call to a function with return variables
$\listsOfVarsMeta=\LST{\varsMeta_1,\dots,\varsMeta_n}$ and the
variables $\listsOfVarsMeta'=\LST{\varsMeta'_1,\dots,\varsMeta'_n}$ in
the caller receiving the return values.
Here, $\storesMeta$ is the store immediately before the call, and
$\storesMeta'$ is the store when the body of the callee stops 
execution.
Formally,
$\callFinSt{\storesMeta}{\storesMeta'}{\listsOfVarsMeta}{\listsOfVarsMeta'}$
is the store $\storesMeta''$ such that for each
$\pdsMeta\in \vars\cup\arrs$, it holds that
\[
  \storesMeta''(\pdsMeta)=
  \begin{cases}
    \storesMeta'(\varsMeta_i) & \mrm{if}~i\in\{1,\dots,n\}
    \,\land\,\pdsMeta=\varsMeta'_i
    \\
    \storesMeta(\pdsMeta) & \mrm{otherwise}
  \end{cases}
\]
and for each $\locsMeta\in \nats$, it holds that
$\storesMeta''(\locsMeta)=\storesMeta'(\locsMeta)$.

\subsubsection*{Definitions for the Array-merging Program and its Verification}{~}\\

We define the expression
$\sepArrFrag{\arrFrag{X}{l_1}{h_1}}
{\arrFrag{Y}{l_2}{h_2}}{\ewhStatesMeta}$ by 

\vspace*{-2ex}
\small
\begin{align*}
  &
    \sepArrFrag{\arrFrag{X}{l_1}{h_1}}{\arrFrag{Y}{l_2}{h_2}}{\ewhStatesMeta}
    % \separated{\arrsMeta_1}{l_1}{h_1}{\arrsMeta_2}{l_2}{h_2}{\ewhStatesMeta}
    ~\defEq~ \\
  &\qquad
    \begin{aligned}
      \exists \locsMeta_1,\locsMeta_2\in \ints:\,
      &
      % \locsMeta_1\ge
      % 0\,\land\,\locsMeta_2\ge 0\,\land\,
      \ewhStatesMeta(\arrsMeta_1)=\locsMeta_1\,\land\,
      \ewhStatesMeta(\arrsMeta_2)=\locsMeta_2\,\land\, 
      (\locsMeta_1+h_1 < \locsMeta_2+l_2 \lor \locsMeta_2+h_2
      < \locsMeta_1+l_1)
    \end{aligned}
\end{align*}
\normalsize

We specify the second loop in the function $\mergeFId$ as

\vspace*{-2ex}
\small
\begin{align*}
  &
    \specMSort(\ewhConfig{\tailLoopStmt{\iVar}{\mVar}}{\ewhStatesMeta}
    {\mergeSortProg}) 
    ~\defEq~ \\
  &
    \quad 
  \begin{aligned}
    \{
    \ewhStatesMeta' \mid\,
    & 
    \ewhStatesMeta'(\iVar) \ge i \,\land\,
    \ewhStatesMeta'(\iVar) = m+1 \,\land\, 
    \ewhStatesMeta'(\jVar) = j \,\land \, 
    \ewhStatesMeta'(\kVar) = k +
    \ewhStatesMeta'(\iVar)-i \,\land  \\
    &
    \preserved{\ewhStatesMeta}{\ewhStatesMeta'}
    {\mVar, \nVar, \TArr, \arrFrag{\TArr}{l}{k-1}}
    % \ewhStatesMeta'(\mVar) = m \,\land\,
    % \ewhStatesMeta'(\nVar) = n \,\land\,
    % \ewhStatesMeta'(\TArr) = \ewhStatesMeta(\TArr)\,\land\, 
    % \\
    % &
    \,\land\,
    \listOfArrFrag{\arrFrag{\SArr}{i}{\ewhStatesMeta'(\iVar)-1}}
    {\ewhStatesMeta}
    =
    \listOfArrFrag{\arrFrag{\TArr}{k}{\ewhStatesMeta'(\kVar)-1}}
    {\ewhStatesMeta'} \,
    % \,\land\, 
    % \listOfArrFrag{\SArr}{l}{n}{\ewhStatesMeta}
    % =
    % \listOfArrFrag{\SArr}{l}{n}{\ewhStatesMeta'}
    % \,\land\, \\
    % &
    % \listOfArrFrag{\TArr}{l}{k-1}{\ewhStatesMeta}
    % =
    % \listOfArrFrag{\TArr}{l}{k-1}{\ewhStatesMeta'} 
    \}
  \end{aligned}
      \\\\[-2.5ex]
  &
    \mrm{if}~
    0\le \lParam \le i \le m < \nVal \,\land\,
    j = \nVal + 1 \,\land\, 
    k = i + j - m - 1 \,\land\,
    \sepArrFrag{\arrFrag{\SArr}{\lVal}{\nVal}}
    {\arrFrag{\TArr}{\lVal}{\nVal}}{\ewhStatesMeta}
    %\BLUE{\separated{\SArr}{\lVal}{\nVal}{\TArr}{\lVal}{\nVal}{\ewhStatesMeta}}
    % \separated{\SArr}{i}{m}{\TArr}{k}{n}{\ewhStatesMeta}
  \\
  &
    \mrm{where}~
    i = \ewhStatesMeta(\iVar)\,\land\,
    j = \ewhStatesMeta(\jVar)\,\land\,
    k = \ewhStatesMeta(\kVar)\,\land\,
    m = \ewhStatesMeta(\mVar)\,\land\,
    n = \ewhStatesMeta(\nVar)
\end{align*}
\normalsize

We specify the third (last) loop of the function $\mergeFId$ as

\vspace*{-2ex}
\small
\begin{align*}
  &
    \specMSort(\ewhConfig{\tailLoopStmt{\jVar}{\nVar}}{\ewhStatesMeta} 
    {\mergeSortProg}) 
    ~\defEq~ \\
  &
    \quad 
    \begin{aligned}
    \{
    \ewhStatesMeta' \mid\,
    &
    \ewhStatesMeta'(\iVar)=i
    % =m+1
    \,\land\,
    \ewhStatesMeta'(\jVar)\ge j \,\land\,
    \ewhStatesMeta'(\jVar)=n+1 \,\land\,
    \ewhStatesMeta'(\kVar)=k+\ewhStatesMeta'(\jVar)-j \,\land\,
    \\
    &
    \preserved{\ewhStatesMeta}{\ewhStatesMeta'}
    {\mVar,\nVar,\TArr,\arrFrag{\TArr}{\lVal}{\kVal-1}}
    % \ewhStatesMeta'(\mVar)=m \,\land\,
    % \ewhStatesMeta'(\nVar)=n \,\land\,
    % \ewhStatesMeta'(\TArr)=\ewhStatesMeta(\TArr)\,\land\,
    % \\
    % &
    \,\land\, 
    \listOfArrFrag{\arrFrag{\SArr}{j}{\ewhStatesMeta'(\jVar)-1}}
    {\ewhStatesMeta}
    =
    \listOfArrFrag{\arrFrag{\TArr}{k}{\ewhStatesMeta'(\kVar)-1}}
    {\ewhStatesMeta'} \, 
    % \,\land\, 
    % &
    % \listOfArrFrag{\SArr}{l}{n}{\ewhStatesMeta}=
    % \listOfArrFrag{\SArr}{l}{n}{\ewhStatesMeta'}
    % \,\land\,
    % \\
    % \listOfArrFrag{\TArr}{l}{k-1}{\ewhStatesMeta}=
    % \listOfArrFrag{\TArr}{l}{k-1}{\ewhStatesMeta'} 
    \}      
  \end{aligned}
  \\\\[-2.5ex]
  &
    \mrm{if}~
    0\le \lParam \,\le m <\, j\le n \,\land\,
    i = m + 1 \,\land\,
    k = i + j - m - 1 \,\land\,
    \sepArrFrag{\arrFrag{\SArr}{\lVal}{\nVal}}
    {\arrFrag{\TArr}{\lVal}{\nVal}}
    {\ewhStatesMeta} 
    %\BLUE{\separated{\SArr}{\lVal}{\nVal}{\TArr}{\lVal}{\nVal}{\ewhStatesMeta}}
   % \separated{\SArr}{j}{n}{\TArr}{k}{n}{\ewhStatesMeta}
  \\
  &
    \mrm{where}~
    i = \ewhStatesMeta(\iVar)\,\land\,
    j = \ewhStatesMeta(\jVar)\,\land\,
    k = \ewhStatesMeta(\kVar)\,\land\,
    m = \ewhStatesMeta(\mVar)\,\land\,
    n = \ewhStatesMeta(\nVar)
\end{align*}
\normalsize

\section{Supplementary Material for \cref{ssec:functional}} \label{asec:subst}

We define the substitution of canonical forms for variables in
expressions, denoted by $\fSubst{\fExprsMeta}{\fVarsMeta}{\cfmsMeta}$,
as follows.

\small
\begin{align*}
  \fSubst{\numeralsMeta}{\fVarsMeta}{\cfmsMeta}
  \,\defEq\,
  &
    \numeralsMeta \\
  \fSubst{\trueLit}{\fVarsMeta}{\cfmsMeta}
  \,\defEq\,
  &
    \trueLit \\
  \fSubst{\falseLit}{\fVarsMeta}{\cfmsMeta}
  \,\defEq\,
  &
    \falseLit \\
  \fSubst{(\fExprsMeta_1\,\op\,\fExprsMeta_2)}{\fVarsMeta}{\cfmsMeta}
  \,\defEq\,
  &
    \fSubst{\fExprsMeta_1}{\fVarsMeta}{\cfmsMeta}\,
    \op\, \fSubst{\fExprsMeta_2}{\fVarsMeta}{\cfmsMeta}
    \quad \mrm{where}~\op\in \{+,-,*,/,=,<\} \\
  \fSubst{(\lnot \exprsMeta_1)}{\fVarsMeta}{\cfmsMeta}
  \,\defEq\,
  &
    \lnot (\fSubst{\exprsMeta_1}{\fVarsMeta}{\cfmsMeta}) \\
  \fSubst{(\FAND{\fExprsMeta_1}{\fExprsMeta_2})}
  {\fVarsMeta}{\cfmsMeta}
  \,\defEq\,
  &
    \FAND
    {\fSubst{\fExprsMeta_1}{\fVarsMeta}{\cfmsMeta}} 
    {\fSubst{\fExprsMeta_2}{\fVarsMeta}{\cfmsMeta}} \\
  \fSubst{(\IFTHENELSE{\fExprsMeta_1}{\fExprsMeta_2}{\fExprsMeta_3})}
  {\fVarsMeta}{\cfmsMeta}
  \,\defEq\,
  &
    \IFTHENELSE
    {\fSubst{\fExprsMeta_1}{\fVarsMeta}{\cfmsMeta}}
    {\fSubst{\fExprsMeta_2}{\fVarsMeta}{\cfmsMeta}}
    {\fSubst{\fExprsMeta_3}{\fVarsMeta}{\cfmsMeta}} \\
  \fSubst{\NIL}{\fVarsMeta}{\cfmsMeta}
  \,\defEq\,
  &
    \NIL \\
  \fSubst{(\CONCAT{\fExprsMeta_1}{\fExprsMeta_2})}
  {\fVarsMeta}{\cfmsMeta}
  \,\defEq\,
  &
  \CONCAT
  {\fSubst{\fExprsMeta_1}{\fVarsMeta}{\cfmsMeta}}
    {\fSubst{\fExprsMeta_2}{\fVarsMeta}{\cfmsMeta}} \\
  \fSubst{(\LCASE{\fExprsMeta_1}{\fExprsMeta_2}{\fExprsMeta_3})}
  {\fVarsMeta}{\cfmsMeta}
  \,\defEq \,
  &
  \LCASE
  {\fSubst{\fExprsMeta_1}{\fVarsMeta}{\cfmsMeta}}
  {\fSubst{\fExprsMeta_2}{\fVarsMeta}{\cfmsMeta}}
    {\fSubst{\fExprsMeta_3}{\fVarsMeta}{\cfmsMeta}} \\
  \fSubst{\fVarsMeta}{\fVarsMeta'}{\cfmsMeta}
  \,\defEq\,
  &
    \begin{cases}
      \cfmsMeta & \mrm{if}~\fVarsMeta=\fVarsMeta' \\
      \fVarsMeta & \mrm{otherwise}
    \end{cases}
  \\
  \fSubst{(\FAPP{\fExprsMeta_1}{\fExprsMeta_2})}
  {\fVarsMeta}{\cfmsMeta}
  \,\defEq\,
  &
  \FAPP
  {\fSubst{\fExprsMeta_1}{\fVarsMeta}{\cfmsMeta}\,}
    {\fSubst{\fExprsMeta_2}{\fVarsMeta}{\cfmsMeta}} \\
  \fSubst
  {(\LAM{\fVarsMeta}{\fExprsMeta_1})}
  {\fVarsMeta'}{\cfmsMeta}
  \,\defEq\, 
  &
    \begin{cases}
      \LAM{\fVarsMeta}{(\fSubst{\fExprsMeta_1}{\fVarsMeta'}{\cfmsMeta})}
      & \mrm{if}~\fVarsMeta\neq \fVarsMeta' \\
      \LAM{\fVarsMeta}{\fExprsMeta_1}
      & \mrm{otherwise}
    \end{cases}
  \\
  \fSubst{
  (\LETRECIN{\fVarsMeta_1}{\LAM{\fVarsMeta_2}{\fExprsMeta_2}}
  {\fExprsMeta_1})}
  {\fVarsMeta}{\cfmsMeta}
  \,\defEq\,
  &
    \begin{cases}
      \LETRECIN{\fVarsMeta_1}
      {\fSubst{(\LAM{\fVarsMeta_2}{\fExprsMeta_2})}{\fVarsMeta}{\cfmsMeta}}
      {\fSubst{\fExprsMeta_1}{\fVarsMeta}{\cfmsMeta}}
      & \mrm{if}~\fVarsMeta\neq \fVarsMeta_1 \\
      \LETRECIN{\fVarsMeta_1}
      {(\fSubst{\LAM{\fVarsMeta_2}{\fExprsMeta_2})}{\fVarsMeta}{\cfmsMeta}}
      {\fExprsMeta_1}
      & \mrm{otherwise}
    \end{cases}
\end{align*}
\normalsize

%%% Local Variables:
%%% mode: latex
%%% TeX-master: "main"
%%% End:

\section{Supplementary Material for \cref{sec:completeness}} \label{asec:proof}

Below, we present the proof of \cref{lem:complete} that is central to
the establishment of the completeness result of the proposed reasoning
technique.

\begin{proof}[of \cref{lem:complete}]  
  We show that for all $\configsMeta$, $\rConfigsMeta$,  
  if $\INFER{\configsMeta}{\rConfigsMeta}{\specsMeta_\star}$,
  then $\rConfigsMeta\in\specsMeta_\star(\configsMeta)$.
  This boils down to showing
  if $\INFER{\configsMeta}{\rConfigsMeta}{\specsMeta_\star}$,
  then 
  $\DERIV{\configsMeta,\rConfigsMeta}$. 
  Below, we give an inductive proof of this statement. 

  Assume $\INFER{\configsMeta}{\rConfigsMeta}{\specsMeta_\star}$. 
  Then, there exist some $m$, $\configsMeta_1$, \dots, $\configsMeta_m$,
  $\rConfigsMeta_1$, \dots, $\rConfigsMeta_m$, and $\rConfigsMeta$, 
  such that 
  $\rConfigFor{\configsMeta_1}{\rConfigsMeta_1}{\specsMeta_\star}$,
  \dots,
  $\rConfigFor{\configsMeta_m}{\rConfigsMeta_m}{\specsMeta_\star}$, 
  and
  \begin{equation}\label{eq:rule_m_completeness}
    \RULE
    {\concl{\configsMeta}{\rConfigsMeta}} 
    {
      \premise{\configsMeta_1}{\rConfigsMeta_1}, \dots,
      \premise{\configsMeta_m}{\rConfigsMeta_m} 
    }
  \end{equation}
  For each $i$, we show that $\DERIV{\configsMeta_i,\rConfigsMeta_i}$ 
  holds by distinguishing between the cases where 
  $\specsMeta_\star(\configsMeta_i)=\rConfigs$
  and $\specsMeta_\star(\configsMeta_i)\neq \rConfigs$. 
  \begin{itemize}
  \item Suppose $\specsMeta_\star(\configsMeta_i)=\rConfigs$. Then
  	we deduce $\INFER{\configsMeta_i}{\rConfigsMeta_i}{\specsMeta_\star}$ 
  	from $\rConfigFor{\configsMeta_i}{\rConfigsMeta_i}{\specsMeta_\star}$.
  	Hence, we have $\DERIV{\configsMeta_i,\rConfigsMeta_i}$ from the 
  	induction hypothesis.  
  \item Suppose $\specsMeta_\star(\configsMeta_i)\neq \rConfigs$. Then
  	we have $\rConfigsMeta_i\in \specsMeta_\star(\configsMeta_i)$ using
  	$\rConfigFor{\configsMeta_i}{\rConfigsMeta_i}{\specsMeta_\star}$.
  	Hence, we have $\DERIV{\configsMeta_i,\rConfigsMeta_i}$ using the
  	definition of $\specsMeta_\star$.
  \end{itemize}
  Ultimately, we have $\DERIV{\configsMeta_i,\rConfigsMeta_i}$ for each
  $i\in \{1,\dots,m\}$, and we obtain $\DERIV{\configsMeta,\rConfigsMeta}$ 
  using \cref{eq:rule_m_completeness}.
  This completes the proof.  \qed
\end{proof}

%%% Local Variables:
%%% mode: latex
%%% TeX-master: "main"
%%% End:

\end{document}

%%% Local Variables:
%%% mode: latex
%%% TeX-master: t
%%% End: